\documentclass[12pt]{article}
\usepackage{amsmath}
\usepackage{graphicx}
\usepackage{enumerate}
\usepackage{natbib}
\usepackage{url} % not crucial - just used below for the URL 

\newcommand{\blind}{1}

\usepackage{subcaption}
\usepackage{amsthm}
\usepackage{verbatim}
\usepackage{multirow}
\usepackage{color}
\usepackage{longtable}
\usepackage{tabularx}
\usepackage{threeparttable}
\usepackage{booktabs}

\usepackage{wanjiestyle}

\newcommand{\Gap}{Gap}

\usepackage{amssymb}

\definecolor{darkgreen}{RGB}{0,100,0}

\begin{document}

\def\spacingset#1{\renewcommand{\baselinestretch}%
{#1}\small\normalsize} \spacingset{1}

%%%%%%%%%%%%%%%%%%%%%%%%%%%%%%%%%%%%%%%%%%%%%%%%%%%%%%%%%%%%%%%%%%%%%%%%%%%%%%

\if1\blind
{
  \title{\bf NetPTR: Optimal Differentially Private Spectral Community Detection on Sparse Networks}
  \author{
    Wanjie Wang\thanks{
    The authors gratefully acknowledge Singapore MOE grant Tier-1-A-8001451-00-00 and NUS Research Scholarship (IRP).}\hspace{.2cm}\\
    Department of Statistics and Data Science, National University of Singapore\\
    Department of Mathematics, National University of Singapore\\
    Tao Shen\hspace{.2cm}\\
    Department of Statistics and Data Science, National University of Singapore}
  \maketitle
} \fi

\if0\blind
{
  \bigskip
  \bigskip
  \bigskip
  \begin{center}
    {\LARGE\bf NetPTR: Optimal Differentially Private Spectral Community Detection on Sparse Networks}
\end{center}
  \medskip
} \fi

\bigskip
\begin{abstract}
Spectral community detection estimates latent labels from the leading eigenspace of a network adjacency matrix, but releasing the resulting labels can disclose sensitive relational information. We consider this problem under differential privacy for both ordinary and bipartite networks. For ordinary networks, the protected unit is a single edge, leading to edge differential privacy (edge-DP). For bipartite networks, the inferential target is the community structure of the left-side nodes, while the protected unit is an entire right-side incidence profile, leading to column-node-DP. 
We propose NetPTR, a private spectral clustering procedure that releases a noisy empirical spectral embedding after a stability test. The algorithm requires perturbation bounds for empirical eigenspaces under neighboring-network changes, which yield computable stability certificates and local sensitivity bounds. For ordinary networks, we establish edge-DP and the error bound under the degree-corrected stochastic blockmodel, which separates the non-private spectral clustering error from the additional privacy-induced error. It therefore guarantees weak consistency in sparse networks and exact recovery in moderate sparse networks. A matching lower bound shows that the required privacy budget is sharp up to logarithmic factors. We further develop a column-node-DP algorithm for bipartite networks and prove consistency under a bipartite degree-corrected block model. Simulations and real-data examples illustrate the resulting privacy--accuracy tradeoff.
\end{abstract}

\noindent%
\vfill

\newpage
\spacingset{1.9} % DON'T change the spacing!

\section{Introduction}\label{sec:intro}
Let $A\in\{0,1\}^{n\times n}$ be the adjacency matrix of an undirected network on a fixed node set $[n]$, and let $\ell\in[K]^n$ denote the latent community labels. Spectral community detection estimates $\ell$ from the leading eigenspace of $A$ via row normalization and clustering \citep{newman2013community, SCORE, hu2024network, Zhao2025}. 
In social network applications, the edge $A_{ij}$ may itself be sensitive. It records whether two individuals maintain a relationship, communication tie, or affiliation. Spectral community labels aggregate the edge pattern into a lower-dimensional structural summary, and this summary contains edge information. Hence, releasing community labels can disclose private edges, especially given auxiliary information about this network. 

Differential privacy (DP) formalizes this concern by imposing stability of the output distribution under neighboring datasets \citep{dwork2014algorithmic}. This neighboring relation determines the unit to protect. For ordinary networks, we use edge adjacency: two networks $A$ and $A'$ are neighboring if they differ in only one undirected edge. An $(\varepsilon, \delta)$-edge-DP community detection procedure must release labels whose distribution is insensitive to the choice of $A$ or $A'$, and hence insensitive to the presence or absence of any single dyadic relation. Therefore, attackers cannot reveal the edge information based on released labels. This privacy requirement is imposed on the released estimator, not on the network itself. 

We also study bipartite community detection. Let $B\in\{0,1\}^{n\times m}$ record edges between left-side nodes $U = [n]$ and right-side units $V = [m]$. The target is to cluster the left-side nodes $U$ using their connection patterns across $V$. 
Such data arise in voting records, user--item interactions, and affiliation networks. In these applications, one unit in $V$ may represent a roll call, item, or organization, and its column in $B$ records the full response or connection profile across $U$. Protecting this profile leads to column-node privacy, which is stronger than protecting a single edge. We therefore define neighboring datasets $B \sim_N B'$, where $B$ and $B'$ differ in one entire column, and we consider the private release of left-side community labels under column-node differential privacy. 

A common strategy for private network analysis is to privatize the whole network and then perform any downstream task on it \citep{karwa2017sharing, mulle2015privacy}. 
This approach is broadly applicable, but it is not tailored to specific tasks. When it comes to community detection, such network-level perturbation suffers too heavy noise and becomes statistically inefficient. Randomized edges may distort degree patterns or weaken community separation, especially in sparse or heterogeneous networks. This motivates a private mechanism that acts directly on the spectral estimator rather than on the full network.

The difficulty is that direct privatization requires a uniform sensitivity bound for the empirical spectral embedding. Differential privacy must hold over all neighboring datasets, whereas spectral community detection theory is typically probabilistic. Standard spectral analyses establish eigenspace stability with high probability under a specific network model. Privacy, however, must also control atypical networks with negligible model probability. Such networks may have small eigengaps, irregular degree behavior, or highly coherent empirical eigenvectors, leading to large sensitivity under neighboring changes. Moreover, unlike mean estimation or regression, spectral clustering has no simple projection onto a convex regular domain: the relevant stability conditions involve global and nonconvex properties of the adjacency matrix. This mismatch between probabilistic spectral theory and deterministic privacy requirements is the central obstacle.  

We develop \emph{NetPTR}, an efficient differentially private spectral community detection approach that privatizes the spectral embedding directly. We develop one-edge perturbation analysis for general networks. Based on this analysis, we construct a a computable stability certificate for the observed network. When the certificate indicates local spectral stability, NetPTR releases a noisy empirical eigenspace with noise calibrated to the certified local sensitivity; otherwise, the method returns a data-independent output. 
This mechanism allows us to connect the differential privacy requirement on all datasets and the spectral community detection theory on regular datasets.

\subsection{Our Contributions}
We propose NetPTR to protect one edge for ordinary networks and Bi-NetPTR to protect one column for bipartite networks. Both approaches are based on a well-designed stability certificate that connects the differential privacy requirement over all datasets and the probabilistic properties on regular datasets. Our contribution are as follows. 

First, our NetPTR privatizes the empirical spectral embedding directly.
NetPTR applies a stability-based release rule to the leading empirical eigenspace and then performs the usual row-normalization and $K$-means steps. We establish theoretical guarantee on its $(\varepsilon, \delta)$-differential privacy. To our knowledge, this is the first estimator-specific edge-DP method for spectral community detection that directly controls the sensitivity of the spectral embedding and comes with statistical recovery guarantees. 

Second, we develop new one-edge perturbation analysis for private spectral clustering. Classical perturbation theories assume arbitrary noise, but not tailored to the perturbation with only one edge. We establish an exact bound that controls the eigenspace perturbation both in spectral norm and in $2 \to \infty$ norm under verifiable conditions involving degree scale, signal and noise eigenvalues, and row-wise incoherence. 
These quantities define a computable stability certificate $\gamma_E$ and a local sensitivity level for noise calibration. 
Hence, this analysis is the key that connects differential privacy to spectral clustering accuracy. 

Third, we establish consistency and rate-optimality of NetPTR under the degree-corrected stochastic blockmodel (DCSBM). The protection of privacy inevitably brings a loss in accuracy, and this accuracy loss must be discussed under model assumptions. Under a regular DCSBM, we find the misclustering rate of NetPTR, which decomposes into the usual non-private spectral clustering error and an additional privacy-induced term. Roughly speaking, when the privacy budget $\varepsilon$ is larger than the reciprocal of average degree, NetPTR achieves weak consistency, and even strong consistency in the moderate sparse networks. We further establish the lower bound for all edge-DP estimators, where no edge-DP estimator can achieve strong consistency when $\varepsilon$ is smaller than this scale. Hence, NetPTR is sharp up to logarithmic factors.

Fourth, we extend the construction to bipartite networks under column-node privacy.
In this setting, the protected unit is a column of the adjacency matrix and the target is the community structure of the left-side nodes. 
Hence, the perturbation geometry is different from edge-DP for ordinary networks. We construct an eigengap-based stability certificate for the left spectral embedding. Based on this stability certificate, we propose Bi-NetPTR, prove $(\varepsilon, \delta)$-column-node DP, and establish consistency; see Theorems~\ref{thm:BiDP}--\ref{thm:BiOptimal}. To our knowledge, this is the first node-DP spectral community detection method with a corresponding statistical analysis for bipartite networks.

Finally, we evaluate our methods on synthetic and real networks. The simulations validate the predicted privacy-utility tradeoff under both DCSBM and Bi-DCSBM, including the effects of privacy budget, network size, degree scale, and the number of right-side nodes. The Flickr contact network and Senate roll-call analysis illustrate that NetPTR can retain interpretable community structure under the corresponding privacy notions.
 
\subsection{Related Literature}
Differential privacy for graph data is commonly studied under edge-level or node-level neighboring relations; see \cite{li2023private} and \cite{mueller2022sok} for surveys. 
Existing work either releases private graph summaries, such as degree
distributions and subgraph counts \citep{hay2009accurate,nguyen2023faster}, or
privatizes the graph through randomized-response and edge-flipping mechanisms
\citep{warner1965randomized,holohan2017optimal}. For a privatized graph, any downstream analysis is private. However, this perturbation is calibrated to the graph representation rather than to a specific inferential target. \cite{mohamed2022differentially} points out that such graph perturbation mechanism cannot outperform other mechanisms, but has a lower computation complexity. \cite{hehir2022consistent} establishes weak consistency of edge-flipped spectral clustering, which requires the degree scale to be at $\sqrt{n}$. Related extensions include edge-flipping spectral clustering for well-clustered graphs \citep{mukherjee2025local}, privacy-integrated graph clustering \citep{mulle2015privacy,nguyen2016detecting} and personalized edge flipping  \citep{zhen2024consistent}. 
More recently, \cite{klopp2026node,marchis2026node} investigate node-private community recovery. 
Research about bipartite networks are much less. \cite{he2024common} has studied the common-neighborhood estimation under edge-DP. Our method instead privatizes the empirical spectral embedding directly and provides theory for both edge-DP ordinary networks and column-node-DP bipartite networks, including sparse networks. 

%The subsequent analyses are private based on these private outputs. However, their perturbation is calibrated to the graph presentation rather than to a specific inferential target, such as community detection. 
%Graph-release mechanisms are broadly applicable because all downstream analyses are private by post-processing, but their perturbation is not adapted to the target estimator. For community detection, this can degrade the low-rank structure that drives spectral recovery, especially in sparse or weak-signal networks. 
%Recent work studies private community recovery under graph perturbation or local privacy, rather than exploiting the perturbation structure of the spectral estimator itself \citep{mohamed2022differentially,hehir2022consistent,mukherjee2025local}. 
%Our method privatizes the spectral clustering estimator directly through a stability certificate derived from network perturbation analysis.

A standard DP release calibrates noise to global sensitivity \citep{dwork2014algorithmic}, which can be conservative when an estimator is stable on typical datasets but unstable on rare inputs.  Propose-Test-Release (PTR) avoids this worst-case calibration by privately testing whether the input is sufficiently far from instability before releasing a noisy statistic \citep{dwork2009differential, brunel2020propose}. 
Recent efficient PTR methods replace the exact distance-to-instability by a computable Lipschitz sub-distance \citep{shen2026efficient}. In our setting, we use this efficient PTR template, but design it to adapt the network setting. The main task is to establish the perturbation results of empirical leading eigenspace, so that we can construct a computable stability certificate that satisfies the requirements of efficient PTR template.

{\it Matrix Perturbation.} Spectral community detection relies on eigenspace perturbation theory \citep{rohe2011spectral, lei2015consistency}. Davis--Kahan bounds control subspace error through an eigengap \citep{daviskahan}, while entrywise and $2\to\infty$ perturbation bounds give the row-wise control needed for exact label recovery \citep{fan2018ell_, tong2025uniform}. These results usually compare a random adjacency matrix with its population mean. Differential privacy requires a different comparison: two neighboring observed networks. For edge-DP, we analyze the one-edge perturbation of the empirical eigenspace and obtain spectral-norm and $2\to\infty$ local sensitivity bounds. This calculation is what makes the stability certificate $\gamma_E$ computable.

\subsection{Organization and Notations}
The remainder of the paper is organized as follows. Section~\ref{sec:prelim} reviews spectral clustering under the network model and the efficient PTR template. Section~\ref{sec:dp} develops the edge-DP NetPTR procedure for ordinary networks, where the consistency and lower-bound guarantee are established in  Section~\ref{sec:optimaldp}. Section~\ref{sec:bipartite} extends the method to bipartite networks under column-node privacy. 
Sections~\ref{sec:simulation} and ~\ref{sec:data} present simulation and real-data studies. Proofs and additional numerical results are deferred to the supplementary material.

For an integer $n$, write $[n]=\{1,\ldots,n\}$. For a matrix $M$, let $\lambda_k(M)$ denote its $k$th largest eigenvalue in magnitude for symmetric $M$, and $\|M\|$, $\|M\|_F$, $\|M\|_{\infty}$ and $\|M\|_{2,\infty}$ its spectral, Frobenius, maximum row-sum, and $2 \to \infty$ norms, respectively. 
%Recall the $2\to \infty$ norm is that $\|M\|_{2,\infty}=\max_{i\in[n]}\|e_i^\top M\|_2$ .
We use $A\sim_E A'$ for edge-adjacent ordinary networks and $B\sim_N B'$ for column-node-adjacent bipartite networks. The notation $c$ and $C$ denotes positive constants whose values may change from line to line. We write $a_n\lesssim b_n$ if $a_n\le Cb_n$ for a constant $C$, and $a_n\asymp b_n$ if both $a_n\lesssim b_n$ and $b_n\lesssim a_n$ hold.

\section{Spectral Community Detection and Efficient PTR}\label{sec:prelim}

\subsection{Network Model and Spectral Community Detection}\label{subsec:dcsbm}
We first introduce the network model and the non-private spectral community detection procedure that will be privatized in later sections. 
While our private procedure does not rely on specific network models, the consistency and optimality analysis require a model.  

Let $A\in\{0,1\}^{n\times n}$ be the adjacency matrix of an undirected graph on the node set $[n]$. We consider the degree-corrected stochastic block model (DCSBM) \citep{DCSBM}. Under DCSBM, $A$ has independent Bernoulli entries where $A_{ij}\sim \mathrm{Bernoulli}(\Omega_{ij})$ for a population matrix $\Omega$, with $A_{ji}=A_{ij}$ and $A_{ii}=0$. The population matrix takes the form 
$\Omega=\Theta\Pi P\Pi^\top\Theta$,
where $\Theta=\mathrm{diag}(\theta_1,\ldots,\theta_n)$ contains node-specific degree parameters, $\Pi\in\{0,1\}^{n\times K}$ is the membership matrix, and $P\in\mathbb{R}^{K\times K}$ is the block connectivity matrix. Let $\ell(i)\in[K]$ denote the community label of node $i$, then $\Pi_{ik}=1\{\ell(i)=k\}$. The statistical target is the label vector $\ell=(\ell(1),\ldots,\ell(n))$ with $K$ being known.

The eigenspace formed by $K$ leading eigenvectors is used to recover $\ell$. 
Let $\Xi_K\in\mathbb{R}^{n\times K}$ collect the eigenvectors of the population matrix $\Omega$ associated with its $K$ largest eigenvalues in magnitude. In \citep{SCORE, lei2015consistency, hu2024network}, it is shown that there exists a nonsingular matrix $H \in \reals^{K \times K}$, so that for node $i$ with $\ell(i) = k$, $e_i^\top \Xi_K=\theta_i e_k^\top H$. Spectral community detection algorithms have been proposed based on this observation in the existing literature, and we consider Algorithm \ref{alg:spectralcluster} in \cite{lei2015consistency} for privacy analysis.

\begin{algorithm}
\caption{Spectral community detection}\label{alg:spectralcluster}
\begin{algorithmic}[1]
\Require Adjacency matrix $A$, number of cluster $K$
\Ensure Clustering centers $\mhat_{[K]}$, and label $\lhat$, and $\Xihat,\Lambdahat$
\State Find eigen-decomposition $A=\Xihat \Lambdahat \Xihat\ic$. 
\State Find the first $K$ eigenvectors matrix $\Xihat_K\in \reals^n$
\State Renormalize $\hat{\Xi}_K$ by rows, where $\rhat_i = {\Xihat_K\ic e_i}/{\|\Xihat_K\ic e_i\|}$
\State Apply K-means on $\{\rhat_{i}\}_{i\in [n]}$ with $K$ clusters, to obtain $\mhat_{[K]}$ and $\lhat$. 
\State Output $\mhat,\lhat, \Xi,\Lambdahat$
\end{algorithmic}
\end{algorithm}

\subsection{Efficient Propose-Test-Release Framework}\label{subsec:ptr}
We first recall the definition of differential privacy. Let $\mathcal A$ be a data space with an adjacency relation $\sim$. A randomized estimator $\hat{\theta}:\mathcal A\to\mathcal O$ is $(\varepsilon,\delta)$-differentially private with respect to $\sim$ if, for all $A,A'\in\mathcal A$ satisfying $A\sim A'$ and all measurable sets $S\subseteq\mathcal O$,
\[
P(\hat{\theta}(A)\in S)
\le
e^\varepsilon P(\hat\theta(A')\in S)+\delta .
\]
Here, $\varepsilon$ and $\delta$ measures the protection on privacy and the adjacency relation specifies the protected unit. 

Propose-test-release (PTR) is designed for estimators with an excessive global sensitivity but a small local sensitivity on a regular set in the data space \citep{dwork2009differential}. 
Given a target sensitivity level $\alpha$, PTR first identifies the datasets $\mathcal{G}_{\alpha}$ on which the estimator is locally $\alpha$-stable. It then privately tests on the distance from $A$ to the unstable region, the complement $\mathcal{G}_{\alpha}^c$. When the test is passed, PTR releases the estimator with a noise calibrated by the local sensitivity $\alpha$; otherwise a data-independent response $\perp$ is released. Hence, the noise level is reduced to the local sensivity $\alpha$.

In the test step of PTR, the exact distance from the observed dataset $A$ to the unstable region $\mathcal{G}_{\alpha}^c$ is required. For spectral estimators, $\mathcal{G}_{\alpha}^c$ is nonconvex and this distance calculation is generally intractable. Efficient PTR (ePTR) replaces the exact distance by a computable Lipschitz sub-distance \citep{shen2026efficient}, which can be regarded as the stability certificate of $A$. 
Its template is in Algorithm \ref{alg:eptr} and the theoretical guarantee is in Theorem \ref{thm:fPTR}.

\begin{algorithm}
\caption{Efficient propose-test-release (ePTR)}\label{alg:eptr}
\begin{algorithmic}[1]
\Require Data $A$, $\widehat\theta$, sub-distance $\gamma$, local sensitivity level $\alpha$, privacy parameters $\varepsilon,\delta$
\Ensure An $(\varepsilon, \delta)$-DP release $\tilde\theta(A)$
\State Compute $M=1+\frac{2}{\varepsilon}\log (\frac{1}{\delta})$ and $\gamma(A)$.
\If{$\gamma(A)>2M$}
 Set $p(A)=1$
\Else 
\[
\mbox{Set}\quad 
p(A)=
\frac{\exp[\varepsilon\{\gamma(A)-M\}/2]}{\exp[\varepsilon\{\gamma(A)-M\}/2]+1}
\]
\EndIf
\State Generate $\zeta \sim \mathcal{N}(0, I)$, and let 
\[
\thetatilde(A) = \left\{\begin{array}{ll}
\mbox{a data-independent response } \bot, \quad & \text{with prob. } 1-p(A),\\
\thetahat(A)+\frac{\alpha}{\varepsilon}\sqrt{2\log(1.25/\delta)} \zeta,   \quad & \text{with prob. } p(A). 
\end{array}\right.
\]
\end{algorithmic}
\end{algorithm}

\begin{thm}
\label{thm:fPTR}
Consider a data space $\mathcal A$ equipped with an adjacency relation $\sim$, and let $\widehat\theta:\mathcal A\to\mathbb R^d$ be an estimator. Fix $\alpha>0$. Suppose there exists a set $\mathcal G\subset\mathcal A$ such that
\[
    \sup_{A\in\mathcal G,\ A'\sim A}
    \|\widehat\theta(A)-\widehat\theta(A')\| \le \alpha .
\]
Let $\gamma:\mathcal A\to\reals_+$ satisfy
    $|\gamma(A)-\gamma(A')|\le 1$ for all 
    $A\sim A'$ and 
    $\{A\in\mathcal A:\gamma(A)>0\}\subset\mathcal G $.
Then the release $\widetilde\theta(A)$ produced by Algorithm~\ref{alg:eptr} is $(\varepsilon,\delta)$-differentially private.
\end{thm}
Theorem~\ref{thm:fPTR} shows that ePTR has a noise calibrated to the local sensitivity $\alpha$, given that the stability certificate $\gamma$ is well-defined. 
For spectral community detection, the derivation of $\alpha$ and $\gamma$ is nontrivial. 
We need $2\to\infty$ perturbation control on the empirical eigenspace $\widehat\Xi_K$ under a one-edge change of the adjacency matrix, which is a non-standard local perturbation analysis.

\section{Edge-Private Spectral Community Detection}\label{sec:dp}

\subsection{One-edge Perturbation and Stability Certificate}
\label{subsec:edge-certificate}
We first define the adjacency relationship for edge-DP, and then derive the global sensitivity and local sensitivity of the empirical eigenspace $\Xihat_K(A)$ in Algorithm~\ref{alg:spectralcluster}. 
For two symmetric adjacency matrices $A,A'\in\{0,1\}^{n\times n}$, define
\[
    d_E(A,A')=\sum_{1\le i<j\le n}1\{A_{ij}\neq A'_{ij}\}.
\]
We write $A\sim_E A'$ if $d_E(A,A')=1$. If an estimator $\tilde{\theta}$ satisfies that for any Borel set $\mathcal{B}$, 
$P(\tilde\theta(A) \in \mathcal{B}) \leq e^{\varepsilon}P(\tilde\theta(A') \in \mathcal{B}) + \delta$ when $A \sim_E A'$, 
then $\tilde\theta$ is called an $(\varepsilon, \delta)$-edge-differentially private (edge-DP) estimator. 

Now we examine the perturbation in the leading eigenvector matrix. 
Let $\Xihat_K(A) \in \reals^{n \times K}$ and $\Xihat_K(A') \in \reals^{n \times K}$ be the matrices consisting of the leading $K$ eigenvectors of $A$ and $A'$, correspondingly. The following proposition bounds both the spectral norm and the $2\to\infty$ norm under one-edge perturbation are controlled up to a rotation, given that $A$ is regular. 
\begin{pro}\label{pro:detailedlinifty}
Let $A,A'$ satisfy $A\sim_E A'$. Suppose there exist constants $a_0>0$, $A_0>0$, and $\theta_0\in(0,1]$ such that
\[
    \|A\|_{\infty}\le (1+a_0)n\theta_0^2,\quad
    \lambda_K(A)\ge a_0n\theta_0^2+3\sqrt 2,
    \quad
    \lambda_{K+1}(A)\le \frac45 a_0n\theta_0^2,\quad
    \|\Xihat_K(A)\|_{2,\infty}\le \frac{A_0}{\sqrt n}.
\]
Then there exists an orthogonal matrix $O\in\mathbb R^{K\times K}$ such that
\[
    \|\Xihat_K(A')-\Xihat_K(A)O\|
    \le
    \frac{5\sqrt 2 A_0}{a_0\theta_0^2 n\sqrt n}
    +
    \frac{50A_0^2}{a_0^2\theta_0^4 n^3},
\quad
    \|\Xihat_K(A')-\Xihat_K(A)O\|_{2,\infty}
    \le U_0,
\]
where
$U_0
=\frac{4\sqrt 2A_0}{a_0\theta_0^2 n\sqrt n}
+\frac{A_0}{a_0\theta_0^2 n\sqrt n}
+\frac{\sqrt 2A_0^2}{a_0\theta_0^2 n^2}
+\frac{5\sqrt 2A_0}{a_0^2\theta_0^4 n^2\sqrt n}
+\frac{50A_0^3}{a_0^2\theta_0^4 n^3\sqrt n}$.
\end{pro}

Proposition~\ref{pro:detailedlinifty} identifies a region on which $\hat{\Xi}_K$ is locally stable under one-edge perturbations. The required conditions have standard spectral interpretations: the maximum degree scale is controlled, the $K$th signal eigenvalue is separated from the noise eigenvalue $\lambda_{K+1}(A)$, and the leading empirical eigenspace is row-wise incoherent. The parameter $\theta_0$ denotes the network sparsity. These conditions motivate the regular set
\[
\begin{aligned}
\mathcal G_E=\Big\{A \in \{0, 1\}^{n \times n}:\;&
\|A\|_{\infty}\le (1+a_0)n\theta_0^2,\quad
\lambda_K(A)\ge a_0n\theta_0^2+3\sqrt 2,\\
&\lambda_{K+1}(A)\le \frac45 a_0n\theta_0^2,\quad
\|\Xihat_K(A)\|_{2,\infty}\le \frac{A_0}{\sqrt n}
\Big\}.
\end{aligned}
\]
On this regular set $\mathcal{G}_E \subset \mathcal{A}$, the local sensitivity level is $\alpha_E
    =
    \sqrt K\Bigl(
    \frac{5\sqrt 2 A_0}{a_0\theta_0^2 n\sqrt n}
    +
    \frac{50A_0^2}{a_0^2\theta_0^4 n^3}
    \Bigr)$ in Proposition \ref{pro:detailedlinifty}. This bound  improves the classical Davis-Kahan bound by an $1/\sqrt{n}$ order. Furthermore, we also establish the sensitivity level in terms of $2\to\infty$ norm, which is required in strong consistency.

The definition of $\mathcal{G}_E$ sheds light on the definition of the stability certificate $\gamma_E(A)$. Here, $\gamma_E(A)$  evaluates the stability of a dataset $A$, where a larger $\gamma_E(A)$ indicates that $A$ is more stable, and $\gamma_E(A) > 0$ indicates that $A \in \mathcal{G}_E$. 
We define $\gamma_E(A)$ via similar terms as $\mathcal{G}_E$ and re-scale these terms to guarantee that $\gamma_E(A)$ is 1-Lipschitz.
\begin{defn}
For any adjacency matrix $A \in \{0, 1\}^{n \times n}$, define the test score $\gamma(A)$ as 
\begin{align}
\gamma_E(A)
=\Bigg[
\min\Bigg\{&
\frac{(1+a_0)n\theta_0^2-\|A\|_\infty}{\sqrt 2},\;
\frac{\lambda_K(A)-a_0n\theta_0^2-3\sqrt 2}{\sqrt 2},\notag\\
&\frac{\frac45a_0n\theta_0^2-\lambda_{K+1}(A)}{\sqrt 2},\;
\frac{\frac{A_0}{\sqrt n}-\|\Xihat_K(A)\|_{2,\infty}}{U_0}
\Bigg\}
\Bigg]_+ .
\label{eqn:DPAgammacomplex}
\end{align}
\end{defn}

Each component of $\gamma_E$ is computable from $A$. The degree and eigenvalue terms are Lipschitz under one-edge changes by elementary norm bounds and Weyl's inequality. The incoherence term is controlled by the $2\to\infty$ perturbation bound in Proposition~\ref{pro:detailedlinifty}. Hence $\gamma_E$ is a valid sub-distance for $\mathcal G_E$, and $\{\gamma_E(A)>0\}\subset \mathcal G_E$.

\subsection{NetPTR for edge-private spectral clustering}
\label{subsec:edge-netptr}
We now propose our NetPTR approach, which proposes a private spectral community labels on social networks. We first release a private spectral space $\tilde{\Xi}_K(A)$ based on the stability certificate $\gamma$. 
It will be a perturbed estimate with noise calibrated by the local sensitivity level $\alpha_E$ in Proposition \ref{pro:detailedlinifty} for large $\gamma_E(A)$; otherwise it is a private release with $\gamma$-dependent probability to be a data-independent vector. 
Then we normalize the rows and apply $k$-means for the labels by Algorithm \ref{alg:spectralcluster}. The theoretical guarantee on privacy is in Theorem \ref{thm:DP}.

\begin{algorithm}[!t]
\caption{NetPTR for edge-private spectral community detection}
\label{alg:spectralDPclusterdetailed}
\begin{algorithmic}[1]
\Require Symmetric adjacency matrix $A$, number of communities $K$, parameters $A_0$, $a_0$ and $\theta_0$, certificate $\gamma_E$, privacy parameters $\varepsilon,\delta$
\Ensure A private label vector $\widetilde\ell$
\State Compute the eigenvalue decomposition $A=\Xihat\Lambdahat\Xihat^\top$.
 Let $\Xihat_K\in\mathbb R^{n\times K}$ collect the leading $K$ eigenvectors in magnitude.
\State Compute $\gamma_E(A)$ as \eqref{eqn:DPAgammacomplex} and $\alpha_E = \sqrt K\Bigl(
    \frac{5\sqrt 2 A_0}{a_0\theta_0^2 n\sqrt n}
    +
    \frac{50A_0^2}{a_0^2\theta_0^4 n^3}
    \Bigr)$.
\State Apply Algorithm~\ref{alg:eptr} to $\Xihat_K$ with sub-distance $\gamma_E$, local sensitivity $\alpha_E$, and privacy parameters $\varepsilon$ and $\delta/2$; denote the output by $\widetilde\Xi_K$.
\If{$\widetilde\Xi_K=\perp$}
    \State Output a fixed data-independent label vector $\widetilde\ell_0$.
\Else
    \State For each $i\in[n]$, set
        $\widetilde r_i=\frac{e_i^\top\widetilde\Xi_K}
        {\|e_i^\top\widetilde\Xi_K\|_2}$.
    \State Apply $K$-means to $\{\widetilde r_i:i\in[n]\}$ to obtain labels $\widetilde\ell$.
    \State Output $\widetilde\ell$.
\EndIf
\end{algorithmic}
\end{algorithm}

\begin{theorem}\label{thm:DP}
For any fixed $a_0>0$, $A_0>0$, and $\theta_0\in(0,1]$, Algorithm~\ref{alg:spectralDPclusterdetailed} is $(\varepsilon,\delta)$-edge-DP.
\end{theorem}

%The proof follows directly from Theorem~\ref{thm:fPTR}. Proposition~\ref{pro:detailedlinifty} verifies the local sensitivity condition on $\mathcal G_E$, and the construction of $\gamma_E$ verifies the sub-distance condition. The row normalization and $K$-means steps are post-processing and therefore do not affect privacy.

It remains to discuss the choice of the density-scale parameter $\theta_0$.
Algorithm~\ref{alg:spectralDPclusterdetailed} uses $\theta_0$ as a density-scale parameter in the stability certificate. In some applications, $\theta_0$ may be chosen from public information or domain knowledge. When this is not available, we propose a private estimate from the maximum degree. Since $\|A\|_\infty$ changes by at most one under edge adjacency, the Laplace mechanism gives
\[
    \widetilde\theta_0(A)
    =
    \left[
    \frac{\|A\|_\infty+\eta}{n}
    \right]_+^{1/2},
    \qquad
    \eta\sim \mathrm{Lap}(1/\varepsilon_1),
\]
as an $(\varepsilon_1,0)$-DP estimate of the network density scale. Running Algorithm~\ref{alg:spectralDPclusterdetailed} with $\widetilde\theta_0(A)$ in place of $\theta_0$ then gives the following composition guarantee.

\begin{cor}\label{cor:private-theta0-edge}
For any fixed $a_0>0$ and $A_0>0$, Algorithm~\ref{alg:spectralDPclusterdetailed} with $\widetilde\theta_0(A)$ is $(\varepsilon+\varepsilon_1,\delta)$ edge differentially private.
\end{cor}
\begin{proof}
The verification is simple, as $|\max_{i}d_i(A)-\max_{i}d_i(A')|\leq 1$, when $d_E(A, A')\leq 1$. By the Laplace mechanism, this is $(\varepsilon_1,0)$-DP. The composition theorem in supplementary materials indicates that $\lhat_{\thetatilde_0}(A)$ is an $(\varepsilon+\varepsilon_1,\delta)$-DP  private estimator. 
\end{proof}

% \subsection{Extension to downstream estimation}
% The $(\varepsilon, \delta)$-DP estimator $\ltilde_{\gamma_E}$ can be used in further estimation problems. For example, the estimation of degree heterogeneity parameters $\theta_i$. 

% Let's start with the non-private estimators $\thetahat$. Let $\mathcal{I}_k(\ell):= \{i \in \mathcal{V}, \ell(i) = k\}$, the set of nodes in community $k$. For node $i$ with $\ell(i) = k$, the degree heterogeneity parameter can be estimated by 
% \[
% \thetahat_i(A,\ell)=\frac{\sum_{j\in \mathcal{I}_{k}(\ell)}A(i,j)}{\sqrt{\sum_{i,j\in \mathcal{I}_{k}(\ell)}A(i,j)}}
% \]
% Suppose $\ell$ is given, it is not difficult to see that 
% \[
% |\thetahat_i(A,\ell)-\thetahat_i(A',\ell)|\leq \frac{2}{\rho n}. 
% \]
% With the Laplace principle, we consider using 
% \[
% \thetatilde_i=\thetahat_i(A,\ltilde)+\frac{2}{n\varepsilon}z,\quad z\sim Lap(1). 
% \]
% \begin{theorem}
% There is a fixed constant $c_0 > 0$. For any $\varepsilon \geq \frac{1}{c_0 n}$, the $\gamma_E$-masked spectral clustering label $\thetatilde$ is $(2\varepsilon,\delta)$-DP, and 
% \[
% \E[|\thetatilde_i-\theta_i|^2]\lesssim \frac{1}{n}+\frac{1}{n^2\varepsilon^2}.
% \]
% Meanwhile, for any $(\varepsilon,\delta)$ estimator $\zeta$ of $\theta_i$ satisfies
% \[
% \E[|\thetatilde_i-\theta_i|]\gtrsim \frac{1}{\sqrt{n}}+\frac{1}{n\varepsilon}.
% \]
% \end{theorem}

\section{Rate Optimality of NetPTR under the DCSBM}
\label{sec:optimaldp}

\subsection{Upper Bound and Consistency}
We now study the error bound of NetPTR. The privacy guarantee in Theorem~\ref{thm:DP} is deterministic and holds for every input graph, whereas error analysis requires a distributional model. We consider the degree-corrected stochastic blockmodel (DCSBM) in the following analysis. We are interested in the statistical error in the non-private estimator, and the additional error introduced to protect privacy. The goal is to minimize the additional error while the privacy is still well-protected. 

We consider a regular DCSBM class, which is commonly used in social network literature \citep{DCSBM, SCORE, hu2024network, shen2025optimal}. 
The communities are non-vanishing, degree parameters comparable to a common scale $\theta_0$, and the block connectivity matrix full rank with bounded spectrum. These conditions ensure that the population eigenspace identifies the community structure and that the non-private estimator enjoys statistical consistency. 

\begin{defn}
Let $\mathbf{1}_n$ denote the $n$-dimensional vector of ones and $e_k$ the $k$th coordinate vector in $\reals^K$. Define the regular label set with respect to a constant $0 < \rho < 1$, 
\[
\mathcal{L}_\rho
:=
\Bigl\{
\Pi\in\{0,1\}^{n\times K}:
\Pi \mathbf{1}_K=\mathbf{1}_n,\;
\mathbf{1}_n\ic \Pi e_k\ge \rho n,\ \forall k\in[K]
\Bigr\}.
\]
Thus each community contains at least $\rho n$ nodes.
\end{defn}

For fixed constants $C_0>0$ and $\rho\in(0,1]$, define the regular population class
\[
\calM_{\theta_0,C_0,\rho}
=
\Bigl\{
\Omega=\Theta \Pi P \Pi\ic \Theta:
\Pi\in \mathcal{L}_{\rho},\ 
\frac{\theta_0}{C_0}<\theta_i<C_0\theta_0,C_0^{-1} < \lambda_K(P) < \lambda_1(P) < C_0, 
\|\Omega\|_\infty\le n\theta_0^2
\Bigr\}.
\]
$\calM_{\theta_0,C_0,\rho}$ contains the population matrices with nonvanishing communities and node-specific degree parameters on a common scale. Hence, it naturally guarantees that the observed network $A$ falls into $\mathcal{G}_E$ with high probability. It is not surprising. To guarantee the consistency of spectral community detection methods, it is indirectly required that the dataset has a stable eigenspace with high probability. 

The following theorem gives the (strong) consistency of NetPTR under this class.
\begin{thm}
\label{thm:optimal}
Fix $C_0>0$, $\rho\in(0,1]$, and integer $K\ge 2$. There exist constants
$a_0>0$ and $A_0>0$, depending only on $(C_0,\rho,K)$, such that the following holds. Let $A$ be generated from the DCSBM with $\Omega\in\calM_{\theta_0,C_0,\rho}$ and $n\theta_0^2>C\log n$. Let $\ltilde_{\gamma_E}$ be the differentially private estimator defined in Algorithm~\ref{alg:spectralDPclusterdetailed}, based on the certificate $\gamma_E$ in \eqref{eqn:DPAgammacomplex}. If
\[
\varepsilon\ge \frac{\log n}{a_0 \theta_0^2 n},
\qquad
\frac{1}{n}\ge \delta\ge \exp\left(-\frac{a_0\theta_0^2 n\varepsilon}{\log n}\right),
\]
then with probability at least $1-O(1/n)$, there exists a permutation $\pi$ of $[K]$ such that
\[
\frac{1}{n}\bigl|\{i:\ltilde_{\gamma_E}(i)\neq \pi\circ \ell(i)\}\bigr|
\le
C\left(
\frac{K}{\theta_0^2 n}
+
\frac{K^2\log(2.5/\delta)}
{\varepsilon^2\theta_0^4 n^2}
\right).
\]
Moreover, if $\theta_0^2\sqrt n\ge 1/a_0$, then $\ltilde_{\gamma_E}=\pi\circ \ell$ with probability at least $1-O(1/n)$.
\end{thm}
Theorem~\ref{thm:optimal} gives both weak and strong consistency guarantees. The weak consistency is guaranteed when the misclustering rate goes to 0. The strong consistency is guaranteed under the additional signal condition $\theta_0^2\sqrt n\ge 1/a_0$, i.e. the average degree is at $\sqrt{n}$. The edge-flipping method needs a degree scale to be linearly dependent on $n$ for strong consistency, which is much stronger \citep{hehir2022consistent}.

The two terms in the upper bound have distinct origins. The first term is the non-private spectral clustering error under the DCSBM, which goes to 0 when the degree scale $n\theta_0^2 \to \infty$. The second term is the additional error induced by the Gaussian perturbation in NetPTR. It decreases as the privacy budget $\varepsilon$ increases and as the degree scale increases. When the privacy budget $\varepsilon \gtrsim \frac{\log n}{n\theta_0^2}$, then the consistency is guaranteed. Thus the theorem gives an explicit privacy--accuracy decomposition.

The same accuracy conclusion continues to hold when the density-scale parameter $\theta_0$ is estimated privately. The next corollary combines the private estimation of $\theta_0$ with the NetPTR release through composition.

\begin{cor}
\label{cor:optimal-private-theta0}
Under the assumptions of Theorem~\ref{thm:optimal}, suppose that
    $\varepsilon_1\ge {\log n}/{(20\theta_0^2 n)}$
and $n\theta_0^2>C\log n$ for a sufficiently large constant $C>0$. Then, with probability at least $1-O(1/n)$, there exists a permutation $\pi$ of $[K]$ such that the release $\ltilde_{\gamma_E}$ by Algorithm~\ref{alg:spectralDPclusterdetailed}, using the private estimate $\thetatilde_0$, satisfies
\[
\frac{1}{n}
\left|\{i:\ltilde_{\gamma_E}(i)\neq \pi\circ \ell(i)\}\right|
\lesssim
\frac{1}{\theta_0^2 n}
+
\frac{\log(2.5/\delta)}
{\varepsilon^2\theta_0^4 n^2}.
\]
Moreover, if $\theta_0^2\sqrt n\ge 1/a_0$, then $\ltilde_{\gamma_E}=\pi\circ \ell$ with probability at least $1-O(1/n)$.
\end{cor}

\subsection{Lower Bound}
We next establish the lower bound for the privacy budget $\varepsilon$ to achieve consistency for any private estimators under DCSBM. The following theorem demonstrates that, even when the degree scale is large and the non-private estimator is consistent, the privacy budget $\varepsilon \lesssim 1/(n\theta_0^2)$ will cause failure in strong consistency for any private estimators. 

\begin{thm}
\label{thm:lowerbound}
For any $\varepsilon<1/(4n\theta_0^2)$, $n\theta_0^2>10\log n$, and $\delta\le \varepsilon^2$, then for any fixed node $i$, no $(\varepsilon,\delta)$-edge DP algorithm satisfies that, 
\[
     P\{\widehat\ell(i)=\ell(i)\}=1-o(1).
\]
\end{thm}

The proof constructs two DCSBM populations that differ only in the label of one fixed node. Under a natural coupling, the sampled networks differ only through edges incident to that node, and their edge distance is of order $n\theta_0^2$ with high probability. Group privacy then prevents any $(\varepsilon,\delta)$-edge DP algorithm with $\varepsilon \lesssim 1/(n\theta_0^2)$ from distinguishing the two labels with probability tending to one.

Theorem~\ref{thm:optimal} shows that our NetPTR approach has strong consistency when $\varepsilon \geq \log(n)/(n\theta_0^2)$. It matches the rate $\varepsilon < 1/(4n\theta_0^2)$ in Theorem~\ref{thm:lowerbound}, up to a logarithmic factor. Hence our NetPTR approach is rate-optimal in privacy budget.

\section{Column-Node-Private Community Detection on Bipartite Networks} \label{sec:bipartite} 

\subsection{Bipartite Spectral Clustering and Column-Node Privacy} \label{subsec:bi-spectral} 
We next extend NetPTR to bipartite networks under a stronger privacy requirement. 
Let $B\in\{0,1\}^{n\times m}$ record interactions between subjects or users $U=[n]$ and items, bills, or other units $V=[m]$. The goal is to cluster the nodes in $U$ from their response across $V$. We consider column-node privacy, which protects the full response profile associated with one item in $V$, rather than protecting only one edge as the ordinary networks. This stronger notion is more practical. For example, a particular question may reveal how all users responded to it. We therefore seek to release the community labels of $U$ while preventing inference about any item-level response profile.

For bipartite community detection on $U$, the spectral approach is to exploit the left singular vectors of $B$, i.e., the leading eigenvectors of $m^{-1}BB\ic$ in magnitude, denoted as $\hat{\Xi}_K$. The non-private spectral clustering algorithm applies row normalization and $K$-means to $\hat{\Xi}_K$; see Algorithm \ref{alg:bispectralcluster}.

\begin{algorithm}
\caption{Bipartite spectral community detection} 
\label{alg:bispectralcluster} 
\begin{algorithmic}[1] 
\Require Bipartite adjacency matrix $B$, number of left communities $K$ 
\Ensure Clustering centers $\mhat_k$, labels $\lhat$, and eigenspace $\Xihat_K$ 
\State Find the singular value decomposition $B=\Xihat\Lambdahat \widehat V\ic$. 
\State Let $\Xihat_K\in\reals^{n\times K}$ collect the first $K$ left singular vectors. 
\State Renormalize  $\rhat_i=\frac{\Xihat_K\ic e_i}{\|\Xihat_K\ic e_i\|}$, $i\in[n]$. 
\State Apply $K$-means to $\{\rhat_i\}_{i\in[n]}$ to obtain $\mhat_{[K]}$ and $\lhat$. 
\State Output $\mhat_{[K]}$, $\lhat$, and $\Xihat_K$. 
\end{algorithmic} 
\end{algorithm} 

We then formalize the column-node privacy. Write $B=[B_1,\ldots,B_m]$, where $B_j\in\{0,1\}^n$ is the column corresponding to node $j\in V$. 
\begin{defn}[Column-node adjacency] 
For two bipartite adjacency matrices $B,B'\in\{0,1\}^{n\times m}$, define \[ d_I(B,B')=\sum_{j=1}^m 1\{B_j\neq B'_j\}. \] We write $B\sim_I B'$ if $d_I(B,B')=1$. 
\end{defn} 
We call an estimator $\tilde{\theta}(B)$ as $(\varepsilon, \delta)$-node DP, if for any Borel set $\mathcal{B}$, $P(\tilde{\theta}(B) \in \mathcal{B}) \leq e^{\varepsilon}P(\tilde{\theta}(B') \in \mathcal{B}) + \delta$. 
It requires the output distribution to be stable when one column of $B$ is replaced. 

\subsection{Bi-NetPTR} 
\label{subsec:binetptr} 
We now construct the column-node-private version of NetPTR. It still requires perturbation analysis and then a proposal of stability certificate. 
Let $Gap(M)=\lambda_K(M)-\lambda_{K+1}(M)$ for a positive semidefinite matrix $M$. By Davis-Kahan theorem, the perturbation in $\Xihat_K(B)$ is related to the eigengap of $BB\ic$. Hence, we use an eigengap certificate. 

Suppose $\theta_0$ denote the universal heterogeneity scale for all left and right side nodes, which is similar as $\theta_0$ in DCSBM. Then $\Gap(BB\ic)$ is at the order of $mn\theta_0^4$ with high probability. We consider the sensitivity level on these networks. 
Define the good set 
\[ \mathcal G_I = \left\{ B:\Gap(BB\ic)>a_0\theta_0^4 nm \right\}, \]
and the sub-distance, i.e., stability certificate, is defined as
\[ 
\gamma_I(B) = \frac{1}{2n} \left\{ \Gap(BB\ic)-a_0\theta_0^4nm \right\}_+ = \frac{m}{2n} \left\{ \Gap\left(\frac1m BB\ic\right)-a_0\theta_0^4 n \right\}_+. \]
Here, $\gamma_I$ has an normalization parameter $(2n)^{-1}$, to ensure it is 1-Lipschitz under column-node adjacency. When $\gamma_I(B)>0$, the eigengap of $m^{-1}BB\ic$ is large enough for Davis--Kahan theorem to control the local sensitivity of $\Xihat_K(B)$. The corresponding local sensitivity level is  $\alpha_I=\frac{4\sqrt{2}}{a_0\theta_0^4m}$, with rigorous derivation in Section \ref{subsec:bi-consistency}.  

Base on the stability certificate and perturbation analysis, we propose Bi-NetPTR as Algorithm \ref{alg:bipartcluster}. For a network $B$ with a large eigengap, we release a perturbed eigenspace $\tilde\Xi_K(B)$, where the noise is calibrated by the local sensitivity level $\alpha_I$. When the network $B$ has a small eigengap, we release a data-independent output. The privacy guarantee is demonstrated in Theorem \ref{thm:BiDP}.
\begin{algorithm}[!t]
\caption{Bi-NetPTR for column-node-private bipartite community detection} \label{alg:bipartcluster} 
\begin{algorithmic}[1] 
\Require Bipartite adjacency matrix $B$, density parameter $\theta_0$, number of left communities $K$, privacy parameters $\varepsilon,\delta$ 
\Ensure A private label vector $\ltilde_{\gamma_I}$ 
\State Find the eigenvalue decomposition  $\frac1m BB\ic=\Xihat\Lambdahat\Xihat\ic $.  
\State Let $\Xihat_K\in\reals^{n\times K}$ collect the first $K$ eigenvectors. \State Compute $\gamma_I(B)$ and $\alpha_I=4\sqrt{2}/(a_0\theta_0^4m)$. 
\State Apply Algorithm~\ref{alg:eptr} to $\Xihat_K$ with sub-distance $\gamma_I$, local sensitivity $\alpha_I$, and privacy parameters $\varepsilon$ and $\delta/2$; denote the output by $\widetilde\Xi_K$. 
\If{$\widetilde\Xi_K=\bot$} 
\State Output a fixed data-independent label vector. 
\Else 
\State Renormalize  $\tilde r_i=\frac{\widetilde\Xi_K\ic e_i}{\|\widetilde\Xi_K\ic e_i\|},\qquad i\in[n]$.  
\State Apply $K$-means to $\{\tilde r_i\}_{i\in[n]}$ to obtain $\ltilde_{\gamma_I}$. 
\State Output $\ltilde_{\gamma_I}$. 
\EndIf 
\end{algorithmic} 
\end{algorithm} 
\begin{thm} \label{thm:BiDP} 
For any fixed $a_0>0$ and $\theta_0>0$, the release $\ltilde_{\gamma_I}$ in Algorithm~\ref{alg:bipartcluster} is $(\varepsilon,\delta)$ column-node differentially private. 
\end{thm} 

Algorithm~\ref{alg:bipartcluster} uses $\theta_0$ as a density-scale input in the stability certificate. In some applications, this parameter can be supplied from public information or domain knowledge. Otherwise, it can be estimated privately. Since $\|B\|_\infty$ changes by at most one under column-node adjacency, the Laplace mechanism gives \begin{equation} \label{eqn:thetatilde_bi} \thetatilde_0(B) = \left[ \frac{\|B\|_\infty}{m} + \frac{\mathrm{Lap}(1/\varepsilon_1)}{m} \right]_+^{1/2}, \end{equation} as an $(\varepsilon_1,0)$-DP estimate. Applying Algorithm~\ref{alg:bipartcluster} with $\thetatilde_0(B)$ gives an $(\varepsilon+\varepsilon_1,\delta)$ column-node-DP estimator by composition. 

\subsection{Consistency under the Bi-DCSBM} 
\label{subsec:bi-consistency} 
We now introduce the bipartite degree-corrected stochastic block model (Bi-DCSBM) for accuracy analysis. 
Conditional on a population matrix $\Omega\in[0,1]^{n\times m}$, the entries of $B$ follow $B_{ij}\sim \mathrm{Bernoulli}(\Omega_{ij})$ independently, $i\in[n]$, $j\in[m]$. $\Omega$ takes the form 
\begin{equation} 
\label{eqn:bisbm} 
\Omega=\Theta\Pi_n P_B\Pi_m\ic\Phi . 
\end{equation} 
Here $\Theta=\mathrm{diag}(\theta_1,\ldots,\theta_n)$ and $\Phi=\mathrm{diag}(\phi_1,\ldots,\phi_m)$ contain degree heterogeneity parameters on the left and right sides, $\Pi_n\in\{0,1\}^{n\times K}$ and $\Pi_m\in\{0,1\}^{m\times J}$ are membership matrices, and $P_B\in\reals^{K\times J}$ is the block connectivity matrix. 
It is connected to the ordinary networks through the population Gram matrix: 
\[ 
\frac1m \E[B]\E[B]\ic = \Theta\Pi_n P\Pi_n\ic\Theta, \qquad P:=\frac1m P_B\Pi_m\ic\Phi^2\Pi_m P_B\ic .
\] 
Thus the left-side spectral problem has the same degree-corrected form as in the ordinary-network case, with an induced block matrix $P$. Under the regularity conditions below, this induced matrix has a nondegenerate $K$-dimensional signal. 

Define the regular label class so that every community on both sides is non-vanishing.
\[ 
\mathcal L_\rho^B := \left\{ (\Pi_n,\Pi_m): \Pi_n\mathbf 1_K=\mathbf 1_n,\ \Pi_m\mathbf 1_J=\mathbf 1_m,\ \mathbf 1_n\ic\Pi_n e_k\ge \rho n,\ \mathbf 1_m\ic\Pi_m e_j\ge \rho m \right\}. 
\] 
For fixed constants $C_0>0$ and $\rho\in(0,1]$, define 
\[ 
\begin{aligned} 
\calM^B_{\theta_0,C_0,\rho} = \Bigl\{ \Omega=\Theta\Pi_n P_B\Pi_m\ic\Phi:\;& (\Pi_n,\Pi_m)\in\mathcal L_\rho^B,\ C_0^{-1}\theta_0<\theta_i,\phi_j\le C_0\theta_0,\\ 
& C_0>\sigma_1(P_B)\ge \sigma_K(P_B)>C_0^{-1},\ \|\Omega\|_\infty\le m\theta_0^2 \Bigr\}. 
\end{aligned} 
\] 
Here $\sigma_k(P_B)$ denotes the $k$th singular value of $P_B$, which ensures the induced signal is full rank. 
The degree parameters on both sides are on the common scale $\theta_0$. 

\begin{thm} \label{thm:BiOptimal} 
For any fixed $C_0>0$, $\rho\in(0,1]$, and integer $K\ge 2$, there exist constants $a_0,c_0>0$ such that the following holds. Let $B$ be generated from the Bi-DCSBM with mean matrix $\Omega\in\calM^B_{\theta_0,C_0,\rho}$. If 
\[
\varepsilon\geq \frac{\log n}{c_0\theta^4_0 m}, \qquad \frac{1}{n}\geq \delta\geq \exp\left( -\frac{a_0\theta^4_0 n\varepsilon}{\log(m\theta_0^4)} \right), 
\]
then with probability at least $1-O(1/m)$, there exists a permutation $\pi$ of $[K]$ such that 
\[ 
\frac{1}{n} \left| \{i:\ltilde_{\gamma_I}(i)\neq \pi\circ \ell(i)\} \right| \lesssim \frac{m+n}{\theta_0^4mn} + \frac{n\log(1.25/\delta)} {\varepsilon^2\theta_0^8m^2}. 
\]
Moreover, there exists a constant $c>0$ such that if 
$\frac{\log n}{\sqrt m} + \frac{\sqrt n\log n\log(1.25/\delta)} {\varepsilon\theta_0^4m} \le c$,  
then with high probability, $\ltilde_{\gamma_I}=\pi\circ\ell$, which means an exact recovery. 
\end{thm} 
The first term in Theorem~\ref{thm:BiOptimal} is the non-private spectral error for estimating the left communities from the bipartite network. The second term is the additional error caused by column-node privacy. The bound decreases as $m$ increases, reflecting that the matrix $m^{-1}BB\ic$ becomes more stable when more right-side profiles are observed. The next result records that the same conclusion holds when $\theta_0$ is privately estimated by \eqref{eqn:thetatilde_bi}. 
\begin{pro} 
\label{pro:Bi2} 
Under the assumptions of Theorem~\ref{thm:BiOptimal}, let Algorithm~\ref{alg:bipartcluster} use $\thetatilde_0(B)$ in \eqref{eqn:thetatilde_bi}. Then the release $\ltilde_{\gamma_I}$ is $(\varepsilon+\varepsilon_1,\delta)$ column-node differentially private. Moreover, if 
\[ 
1\geq \varepsilon_1\geq \frac{\log m}{a_0\theta^2_0 m}, \qquad \varepsilon\geq \frac{\log n}{c_0\theta^4_0 m}, \qquad \frac{1}{n}\geq \delta\geq \exp\left( -\frac{a_0\theta^4_0 n\varepsilon}{\log(m\theta_0^4)} \right), \] 
then the conclusion of Theorem~\ref{thm:BiOptimal} continues to hold. 
\end{pro}

\section{Simulation}\label{sec:simulation}

\subsection{Edge-Private Community Detection}\label{subsec:sim_symmetric}
We first evaluate the edge-private NetPTR procedure under the DCSBM. In each experiment, we fix $K=2$ and assign exactly $n/2$ nodes to each community. The block connectivity matrix $P \in \reals^{2 \times 2}$ has $p_{in}$ on diagonals and $p_{out}$ on diagonals. The population matrix $\Omega=\Theta\Pi P\Pi^\top\Theta$, where $\Theta$ and $\Pi$ are randomly generated from a given distribution in each repetition. We then sample $A$ by $A_{ij}\sim{\rm Bernoulli}(\Omega_{ij})$ independently for $i<j$. It is symmetrized to obtain an undirected adjacency matrix. 

We consider two scenarios. The first is a regular scenario with $(p_{\rm in},p_{\rm out})=(0.4,0.1)$ and $\theta_i\sim{\rm Unif}(0.1,0.5)$. The second is a heterogeneous scenario with $(p_{\rm in},p_{\rm out})=(0.9,0.3)$, where $\theta_i\sim{\rm Unif}(0.1,0.5)$ with probability $0.6$ and $\theta_i\sim{\rm Unif}(0.01,0.05)$ with probability $0.4$. Hence, it contains a large fraction of low-degree nodes, which is more challenging for private spectral clustering. We examine the accuracy of NetPTR in the two scenarios versus sample size $n$, privacy budget $\varepsilon$, and private $\theta_0$. The accuracy is evaluated by the Hamming distance between the estimated label and ground truth. 

{\bf Sample size and privacy budget.}
We first study NetPTR as a standalone private clustering method. We vary the sample size $n\in\{5000,10000,15000,\cdots,50000\}$, the privacy budget $\varepsilon\in\{0.5,0.8,1.0\}$, and fix $\delta=0.01$. We compare two versions of Algorithm~\ref{alg:spectralDPclusterdetailed}: a non-private estimate $\hat\theta_0$, and a fully private version using $\tilde\theta_0$ with $\varepsilon_1=0.2$. The non-private community detection algorithm is included as a comparison. 
Each configuration is repeated 50 times.

Figure~\ref{fig:exp1_err} summarizes the mean community detection error rate versus $n$ for each $\varepsilon$. To protect the privacy, NetPTR suffers a larger clustering error rate than the non-private algorithm. This error rate decreases as the sample size increases or the privacy budget increases. In the regular scenario, NetPTR achieves an error rate at 0 at moderate sample size $n \geq 20,000$. With private $\tilde\theta_0$, the error rate of NetPTR is slightly larger, but still achieves 0 when $n \geq 30,000$. In the heterogeneous scenario with extreme sparsity, even the non-private method has a larger clustering erro rate. NetPTR introduces a mild additional error rate at around $0.1$ in this scenario. It indicates the robustness of NetPTR.

\begin{figure}[!htbp]
  \centering
  \begin{subfigure}[t]{0.49\textwidth}
    \centering
    \includegraphics[width=\textwidth]{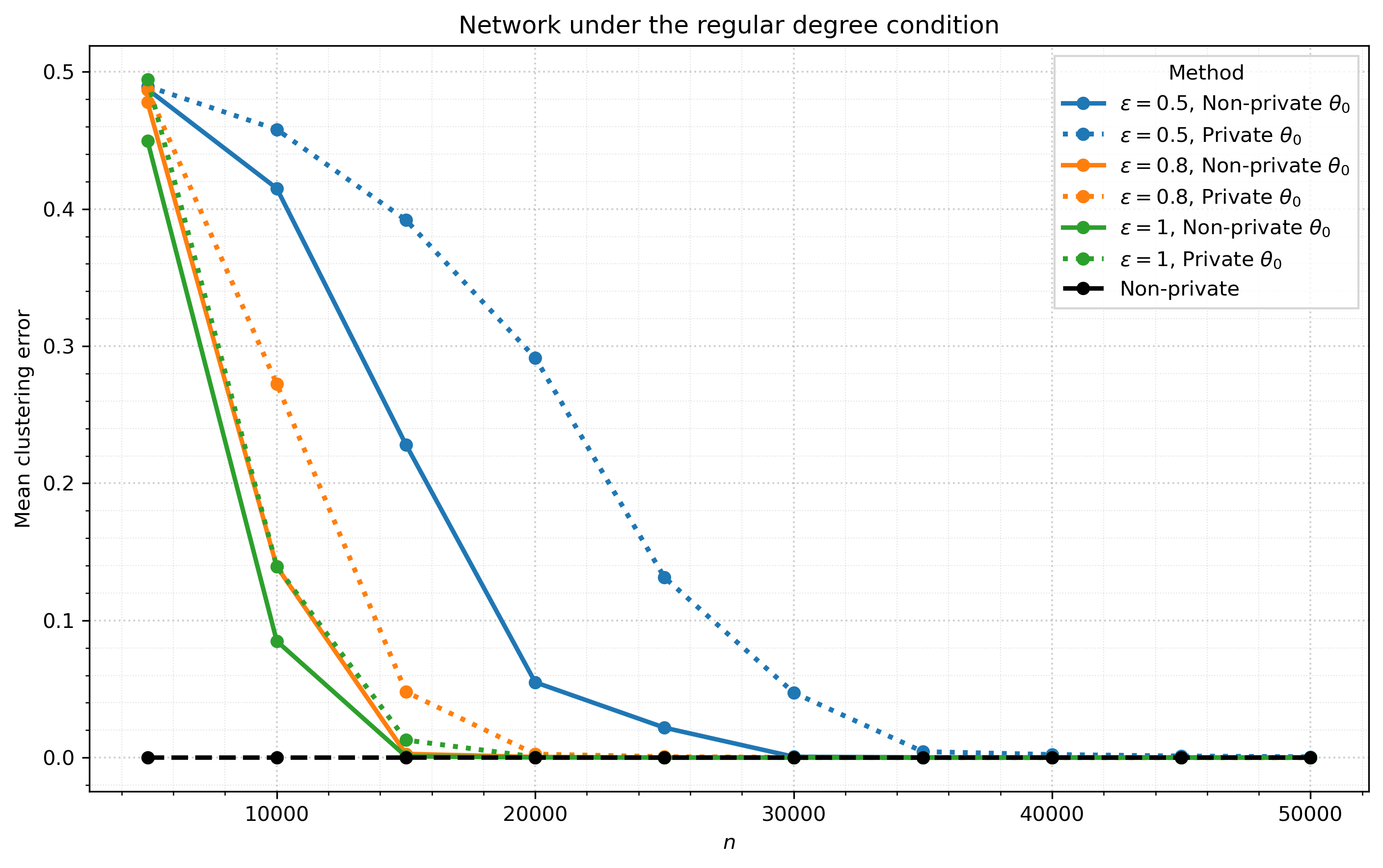}
  \end{subfigure}\hfill
  \begin{subfigure}[t]{0.49\textwidth}
    \centering
    \includegraphics[width=\textwidth]{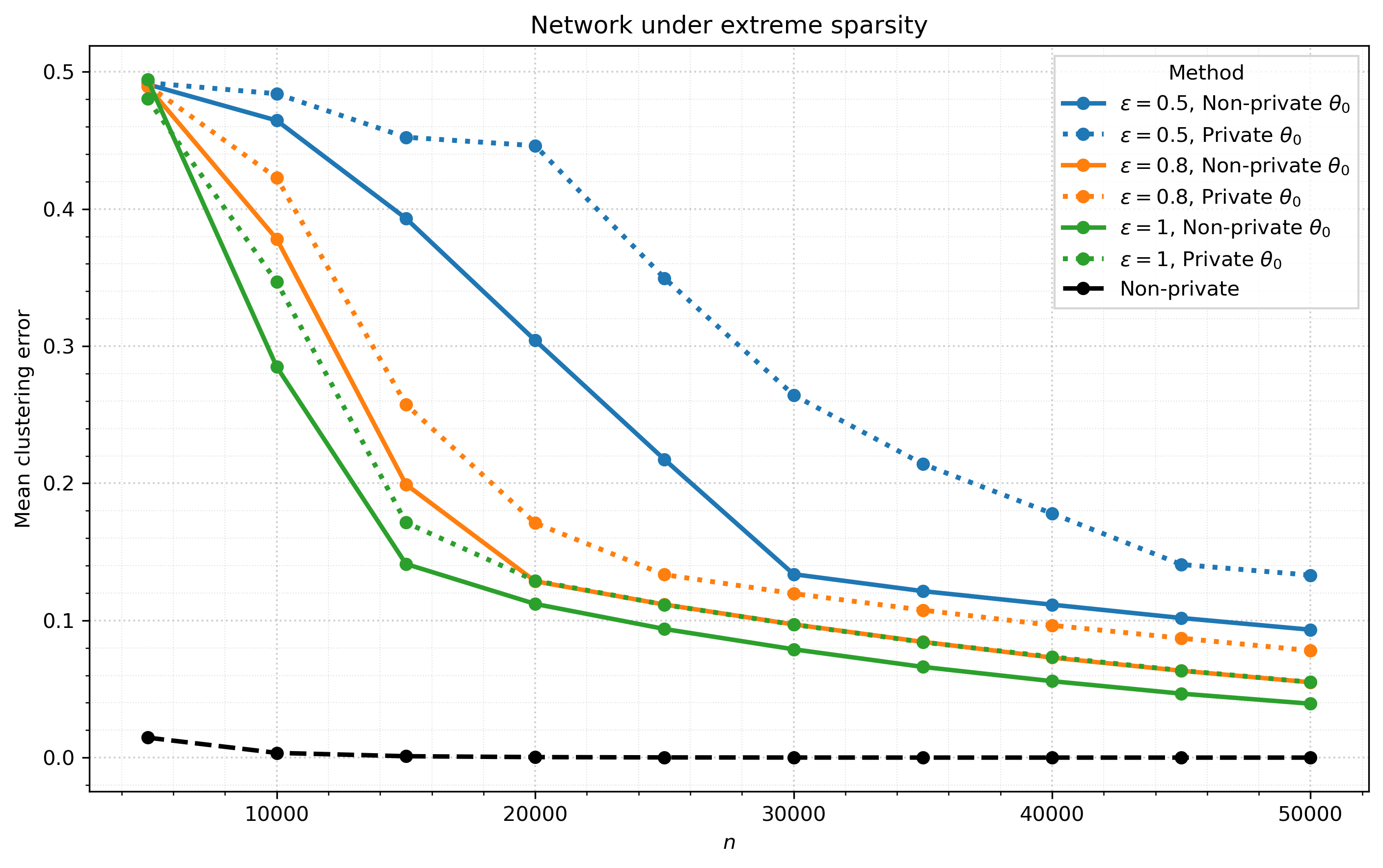}
  \end{subfigure}

  \caption{
  Mean clustering error rate versus network size $n$, with privacy budget $\varepsilon \in \{0.5,0.8,1.0\}$ for NetPTR with non-private $\hat\theta_0$ (solid line) and private $\tilde\theta_0$ (dotted line).
  }
  \label{fig:exp1_err}
\end{figure}

{\bf Comparison with edge flipping.}
We compare our NetPTR with the classical edge-flipping baseline (EdgeFlip, implemented following the procedure in \cite{mohamed2022differentially}). Fix $n=20000$ and vary the privacy budget $\log_2(\varepsilon) \in \{-1,-0.75,\cdots,0.75,1\}$ and $\delta = 0.01$. For each $\varepsilon$, we run NetPTR with known $\theta_0$ $(\varepsilon_0 = \varepsilon, \varepsilon_1 = 0)$ and a private estimate $\tilde{\theta}_0$ $(\varepsilon_0 = \varepsilon-0.2, \varepsilon_1 = 0.2)$ respectively, as well as EdgeFlip. We summarize the mean clustering error over 50 independent replications in Figure \ref{fig:exp2_edgeflip}. It can be seen that NetPTR, whether with non-private $\hat\theta_0$ or private $\tilde\theta_0$, improves substantially faster as the privacy budget $\varepsilon$ increases in both scenarios.
\begin{figure}[!htbp]
  \begin{subfigure}[t]{0.49\textwidth}
    \centering
    \includegraphics[width=\textwidth]{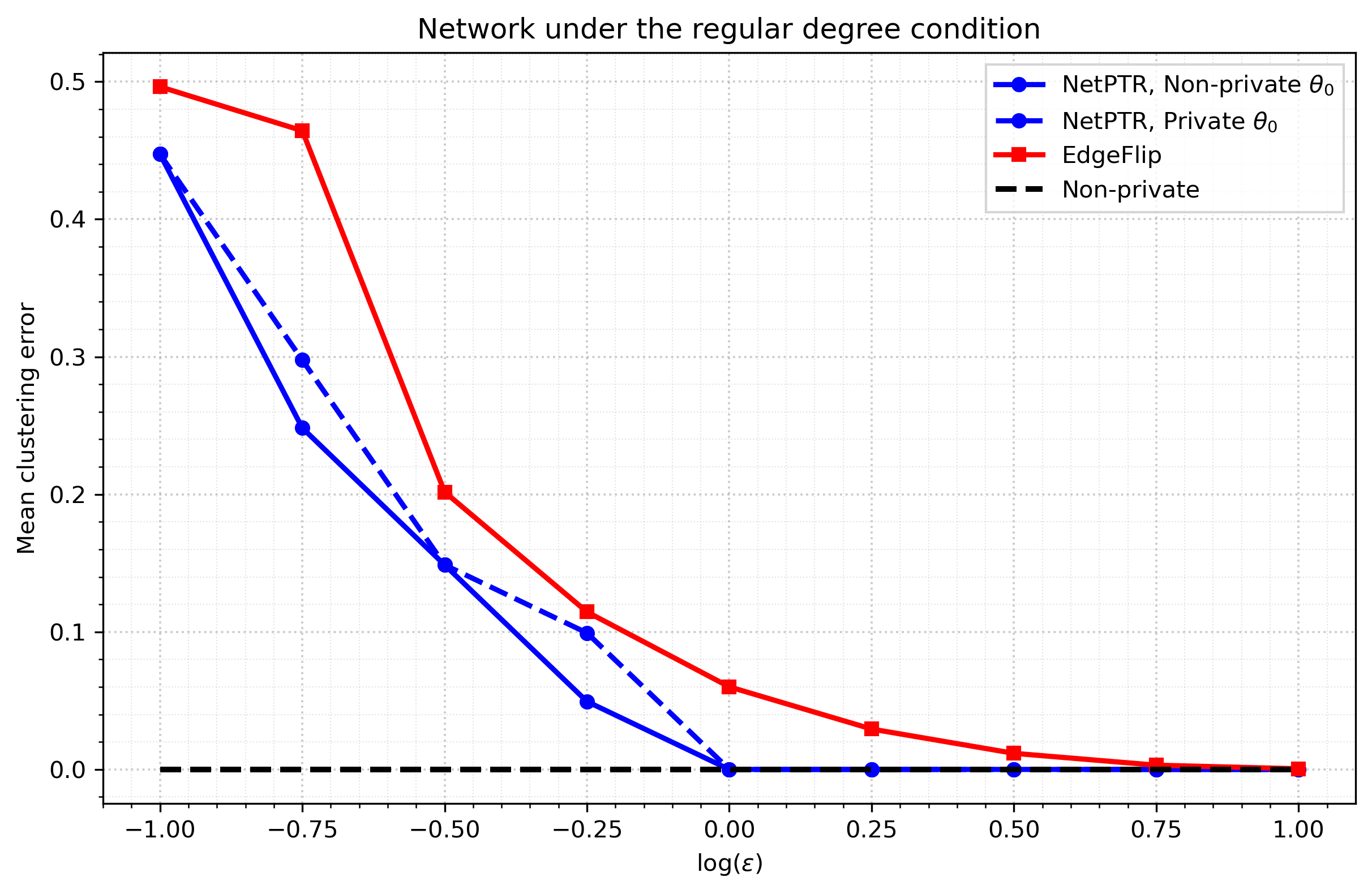}
  \end{subfigure}\hfill
  \begin{subfigure}[t]{0.49\textwidth}
    \centering
    \includegraphics[width=\textwidth]{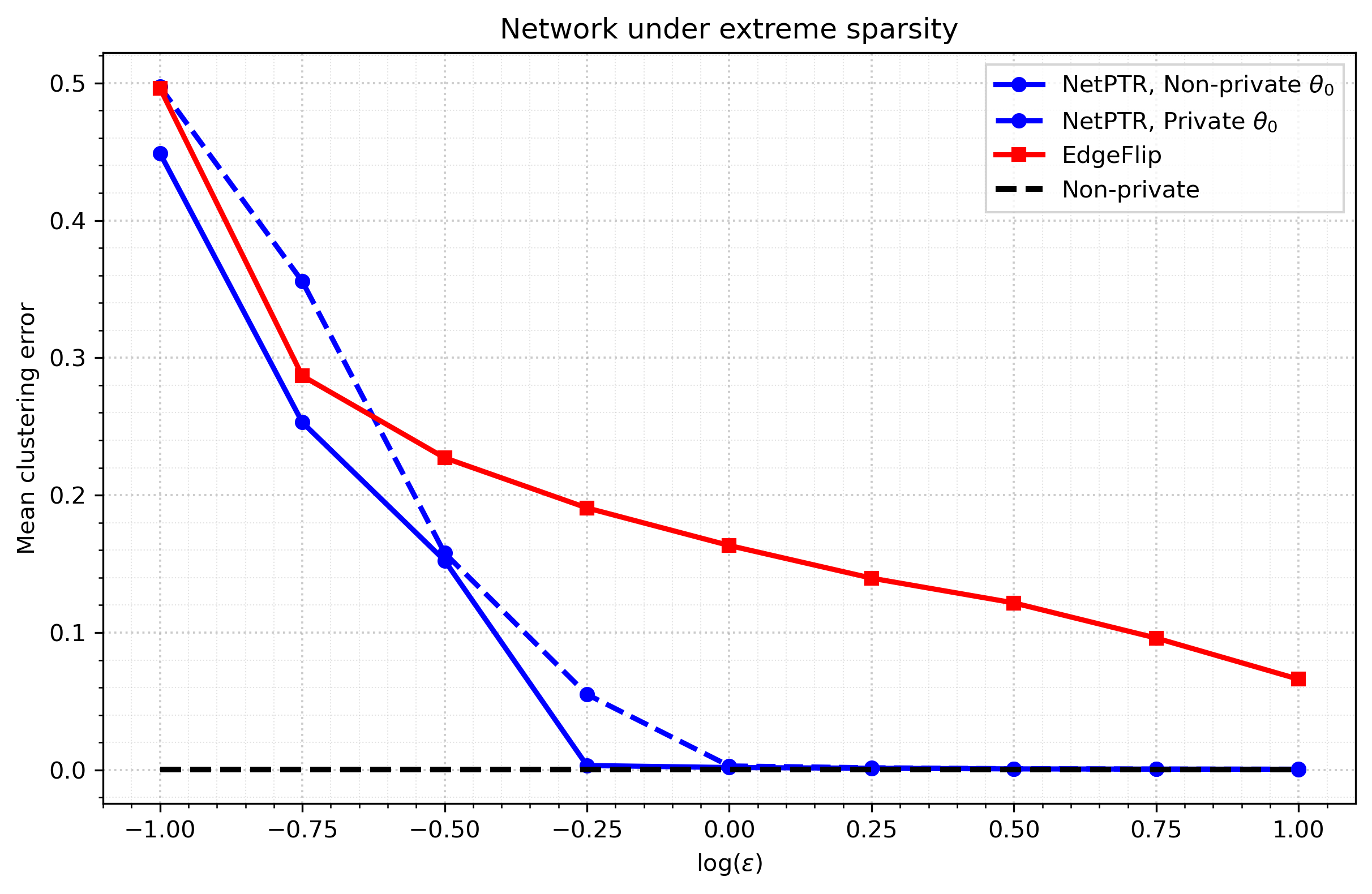}
  \end{subfigure}
  \caption{Mean clustering error versus $\log(\varepsilon)$ for NetPTR and EdgeFlip.}
  \label{fig:exp2_edgeflip}
\end{figure}

{\bf Network Sparsity.} 
We examine the privacy-utility behavior with varying $\theta_0$ levels. Let $\theta_i\sim \mathrm{Unif}(a,a + 0.4)$ for the regular scenario and $\theta_i\sim 0.6 \mathrm{Unif}(a,a + 0.3) + 0.4 \mathrm{Unif}(0.1a,0.1a + 0.03)$ for the scenario with extreme sparsity. The parameter $\theta_0 = \sqrt{\|\Omega\|_{\infty}/n}$ is decided by $a$. We vary $a$ to figure out its values that $\theta_0$ ranges from .10 to .50 by a step size of $.05$ in the regular setting, and $\theta_0$ ranges from .175 to .550 by a step size of $.0625$ for the scenario with extreme sparsity. Fix $n=20000$, $\varepsilon = 0.8$ and $\delta = 0.01$.

\begin{figure}[!htbp]
  \centering
 \begin{subfigure}[htbp]{0.49\textwidth}
    \centering
    \includegraphics[width=\textwidth]{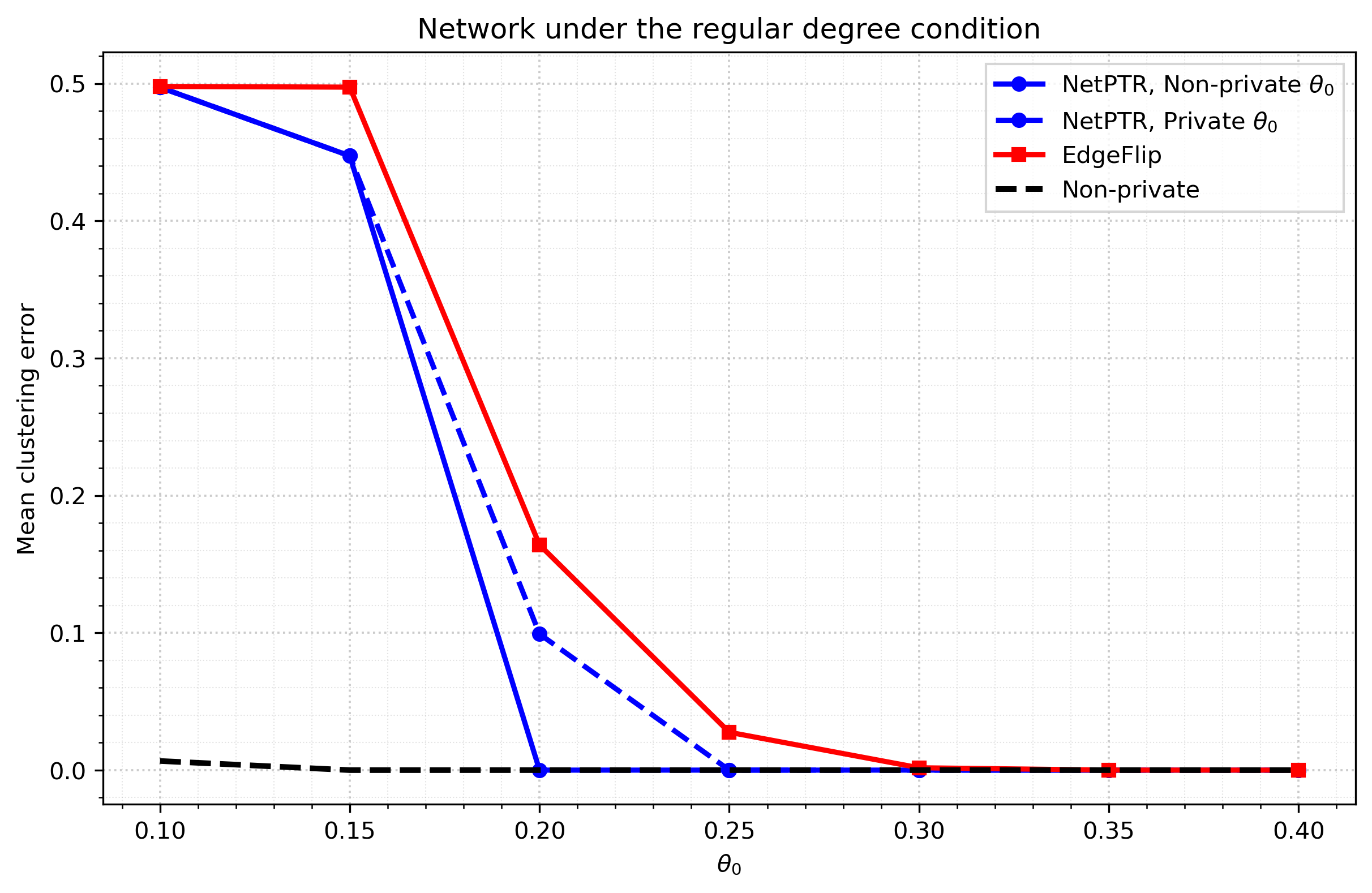}
  \end{subfigure}\hfill
  \begin{subfigure}[htbp]{0.49\textwidth}
    \centering
    \includegraphics[width=\textwidth]{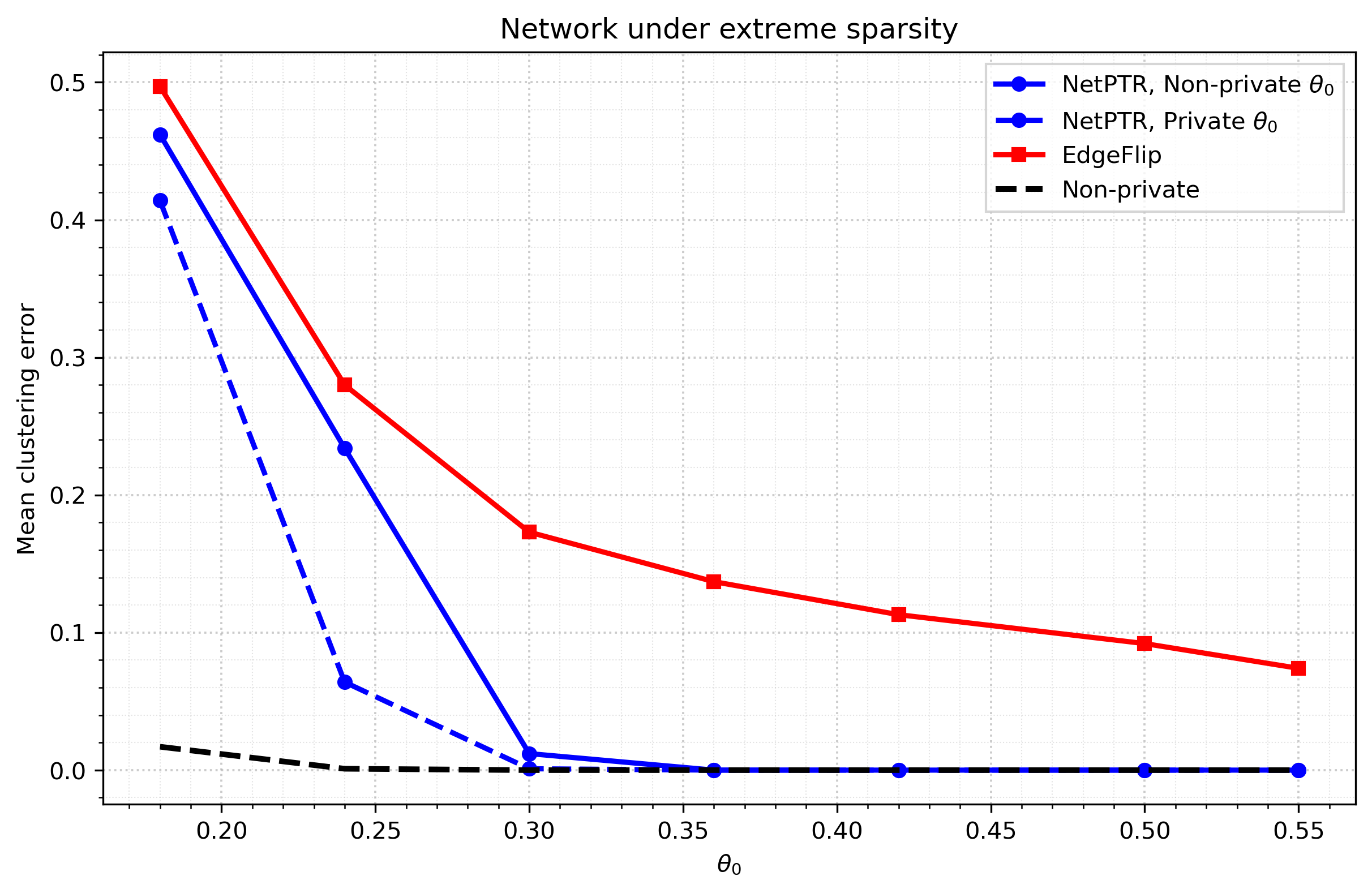}
  \end{subfigure}
  \caption{Mean clustering error versus the effective degree scale $\theta_0$ for different mechanisms.}
  \label{fig:exp3_theta0}
\end{figure}

The mean clustering errors over 50 repetitions are summarized in Figure~\ref{fig:exp3_theta0}. For both scenarios, the non-private spectral clustering baseline is nearly exact across varying $\theta_0$, so the separation between curves mainly reflects privacy-induced perturbations. As $\theta_0$ increases, the performance of both private mechanisms improves. NetPTR improves more rapidly and achieves near-exact recovery at moderate $\theta_0$, whereas EdgeFlip remains at an error rate $0.2$ in the scenario with extreme sparsity with large $\theta_0$.
It is consistent with Theorem~\ref{thm:optimal}, where NetPTR has an error bound with respect to the network density, but EdgeFlip suffers in this degree heterogeneous setting.

\subsection{Bipartite-network experiments under Bi-DCSBM}
\label{subsec:sim_bipartite}
We evaluate Bi-NetPTR in Section \ref{sec:bipartite}. 
Consider $n = 800$ left nodes and $m = 8000$ right nodes. Fix $K = 2$ left-communities and $J = 2$ right-communities, and assign exactly $n/2$ left nodes and $m/2$ right nodes to each community. 
Generate the degree parameters independently as $\theta_i, \phi_j \sim \mathrm{Unif}(0.7,1.0)$. The adjacency matrix follows that $B_{ij}\sim \mathrm{Bernoulli}(\theta_i \phi_j p_{\ell_1(i), \ell_2(j)})$, where 
$p_{k,l} = 0.7*I\{k = l\} + 0.1*\{k \neq l\}$.
For each generated $B$, we apply the non-private bipartite spectral clustering method as Algorithm \ref{alg:bispectralcluster}, as well as the Bi-NetPTR method with $\varepsilon$ and $\delta=0.01$.

{\bf Privacy budget.} 
We first study how the clustering accuracy changes with the privacy budget $\varepsilon$.
We vary
\(
\varepsilon
\in
\{0.5,1.0,2.0,4.0\},
\)
and consider three privacy-budget splits for Bi-NetPTR,
\(
\varepsilon_1\in\{0,0.1,0.2\},
\)
where $\varepsilon_1>0$ corresponds to using a private estimate of $\theta_0$.

The upper panel of Table~\ref{tab:bipartite_sim} reports the mean clustering error, standard deviation, and empirical release frequency among 50 repetitions.
The non-private method achieves zero error in all repetitions and is omitted from the table. 
As $\varepsilon$ increases, the clustering error decreases to 0 no matter what split we use. For a small privacy budget $\varepsilon = 0.5$, allocating part of the budget to $\theta_0$ leads to a lower release frequency and hence larger error. When $\varepsilon \geq 1.0$, the release probability increases to 1 quickly. Hence, the private estimation of $\theta_0$ is practically feasible.

{\bf Varying the number of right-side nodes.}
We fix $n=800$ and $\varepsilon=3.0$, and vary the number of right-side nodes over $
m\in\{1000,1500,2000,3000,$
$4000,5000,6000,7000\}.$
We consider $\varepsilon_1\in\{0,0.1,0.2\}$. The results over 50 repetitions are recorded in the lower panel of Table~\ref{tab:bipartite_sim}. 

\begin{table}[!t]
\centering
\scriptsize
\resizebox{\textwidth}{!}{
\begin{tabular}{ll|cc|cc|cc}
\hline
Experiment & Setting
& \multicolumn{2}{c|}{Bi-NetPTR, $\varepsilon_1=0$}
& \multicolumn{2}{c|}{Bi-NetPTR, $\varepsilon_1=0.1$}
& \multicolumn{2}{c}{Bi-NetPTR, $\varepsilon_1=0.2$} \\
\cline{3-8}
& & Error (SD) & Rel. & Error (SD) & Rel. & Error (SD) & Rel. \\
\hline
\multicolumn{8}{l}{\textit{Varying $\varepsilon$ with $(n,m)=(800,8000)$}} \\
\hline
& $\varepsilon=0.5$ & 0.443 (0.043) & 0.60 & 0.479 (0.027) & 0.10 & 0.484 (0.009) & 0.00 \\
& $\varepsilon=1.0$ & 0.298 (0.014) & 1.00 & 0.334 (0.065) & 0.87 & 0.363 (0.062) & 0.80 \\
& $\varepsilon=2.0$ & 0.145 (0.011) & 1.00 & 0.142 (0.012) & 1.00 & 0.161 (0.014) & 1.00 \\
& $\varepsilon=4.0$ & 0.021 (0.005) & 1.00 & 0.017 (0.005) & 1.00 & 0.018 (0.004) & 1.00 \\
\hline
\multicolumn{8}{l}{\textit{Varying $m$ with $n=800$ and $\varepsilon=3.0$}} \\
\hline
& $m=1000$ & 0.470 (0.032) & 0.27 & 0.488 (0.008) & 0.03 & 0.479 (0.017) & 0.07 \\
& $m=2000$ & 0.314 (0.015) & 1.00 & 0.442 (0.079) & 0.27 & 0.452 (0.067) & 0.20 \\
& $m=3000$ & 0.236 (0.015) & 1.00 & 0.261 (0.092) & 0.87 & 0.282 (0.098) & 0.80 \\
& $m=4000$ & 0.166 (0.013) & 1.00 & 0.162 (0.015) & 1.00 & 0.168 (0.012) & 1.00 \\
& $m=5000$ & 0.119 (0.013) & 1.00 & 0.108 (0.012) & 1.00 & 0.119 (0.012) & 1.00 \\
& $m=6000$ & 0.075 (0.010) & 1.00 & 0.075 (0.008) & 1.00 & 0.079 (0.009) & 1.00 \\
& $m=7000$ & 0.047 (0.007) & 1.00 & 0.046 (0.008) & 1.00 & 0.050 (0.009) & 1.00 \\
\hline
\end{tabular}
}
\caption{Simulation results under Bi-DCSBM. 
Each error entry reports mean clustering error with standard deviation in parentheses; ``Rel.'' denotes the empirical release frequency.}
\label{tab:bipartite_sim}
\end{table}

With a larger $m$, the perturbation of one column has a smaller effect on the eigengap of $m^{-1}BB^\top$, and all Bi-NetPTR methods improve. Using a private $\tilde\theta_0$, the release probability is small when $m \leq 2000$, and it turns to one when $m \geq 3000$. Overall, these results support the practical effectiveness of the Bi-NetPTR procedure.

\section{Real Data Analysis}
\label{sec:data}

\subsection{Flickr Contact Network}
The Flickr network comes from the Social Computing Data Repository (\url{https://datasets.syr.edu/Flickr/}), where each node represents a user and each edge represents the friendship between the two users. 
In such social networks, a user may wish to keep certain connections private, for example when the connection is with a controversial or reputationally sensitive user. This makes the dataset a natural example for edge-private community detection.
We remove the nodes with degree smaller than $100$. In the processed network, there are $n=26082$ nodes and $4416654$ undirected edges. 

We apply NetPTR and EdgeFlip on this network for $\varepsilon \in \{2.0, 2.5, 3.0, 3.5, 4.0\}$ and $\delta = 0.01$ for 20 independent repetitions. 
To each repetition, we randomly choose $10\%$ nodes for private estimate. For NetPTR, we consider private $\tilde{\theta}_0$ with $\varepsilon_1 = 0.2$ and $\varepsilon_1 = 0.5$. The parameters are chosen as $a_0 = .1$ and $A_0 = 50$. 

We assess utility by comparing each private clustering result with the non-private spectral clustering partition using the adjusted Rand index (ARI). A higher ARI indicates a better alignment between the two estimates. 
Table~\ref{tab:flickr_summary} reports the mean ARI values across 20 replications. As expected, the ARI of both methods improves as the privacy budget $\varepsilon$ increases. EdgeFlip and NetPTR have a similar performance for small $\varepsilon = 2.0$, but NetPTR clearly outperforms EdgeFlip when $\varepsilon \ge 2.5$, where the gap widens steadily as $\varepsilon$ increases. In addition, EdgeFlip is computationally less appealing because it must explicitly perturb ${n \choose 2}$ possible edges, while NetPTR adds noise to $n \times K$ entries.  
Overall, NetPTR protects the privacy, outperforms other methods, and is computationally efficient.

%This reference partition is meaningful for this dataset: the appendix shows that the induced binary split agrees closely with the auxiliary Flickr group structure, in the sense that, for 183 of the 195 groups, an overwhelming majority of members lie on the same side of the partition. 

\begin{table}[!htbp]
\centering
\small
\setlength{\tabcolsep}{4pt}
\begin{tabular}{cccc}
\toprule
$\varepsilon$ & EdgeFlip & NetPTR ($\varepsilon_1=0.2$) & NetPTR ($\varepsilon_1=0.5$) \\
\midrule
2.0 & 0.9016 (0.0009) & 0.9020 (0.0015) & 0.9038 (0.0017) \\
2.5 & 0.9055 (0.0007) & 0.9219 (0.0011) & 0.9213 (0.0013) \\
3.0 & 0.9114 (0.0008) & 0.9355 (0.0012) & 0.9360 (0.0009) \\
3.5 & 0.9206 (0.0009) & 0.9455 (0.0010) & 0.9461 (0.0006) \\
4.0 & 0.9331 (0.0004) & 0.9540 (0.0013) & 0.9546 (0.0009) \\
\bottomrule
\end{tabular}
\caption{Mean ARI on the Flickr network, with standard errors shown in parentheses.}
\label{tab:flickr_summary}
\end{table}

\subsection{Senate Roll-call Voting}
The roll-call voting records from the 109th U.S. Senate are represented as a bipartite network connecting senators to roll-call votes \citep{jackman2015package, lewis2019voteview}. After preprocessing, the data contains $n=102$ legislators as left-side nodes and $m=406$ roll calls as right-side nodes. Each edge indicates that the legislator votes ``yes" to the corresponding roll call. Each legislator is encoded with their party, which is the ground truth. The dataset and processing code can be found at \url{https://tinyurl.com/NetPTR-SC}.

We apply non-private and private spectral clustering with $K=2$ to the legislators. The private estimates include EdgeFlip, Bi-NetPTR with $\varepsilon_1=0$, and Bi-NetPTR with $\varepsilon_1=0.5$. 
The edge-flipping mechanism was designed for edge-DP, and here we extend it into the node-level privacy via $\varepsilon_{edge} = \varepsilon/m$. 
The left panel in Figure~\ref{fig:senate_realdata} reports the ARI between the private outputs and the non-private benchmark versus the privacy budget $\varepsilon$, where a larger ARI indicates a higher agreement. 
As $\varepsilon$ increases, both Bi-NetPTR variants have larger ARI scores, but EdgeFlip remains at 0 since $\varepsilon_{edge}$ remains large even at moderate $\varepsilon$. This difference explicitly explains the difficulty in node-level privacy. It suggests to consider a medium privacy regime. 

At $\varepsilon = 8$, the right panel in Figure~\ref{fig:senate_realdata} shows the party composition within each recovered cluster for the non-private method and for one replicate of Bi-NetPTR with $\varepsilon_1=0.5$. 
The non-private method exhibits an almost perfect partition between the Democrats and Republicans, where the accuracy is $0.98$. The private estimate, Bi-NetPTR, largely preserves this structure with a better privacy protection, at an reduced accuracy of $0.95$ for Democrats, but the same accuracy for Republicans. Therefore, Bi-NetPTR still captures the structural information with effective privacy protection. 

\begin{figure}[htbp]
\centering
\begin{minipage}[htbp]{0.49\textwidth}
    \centering
    \includegraphics[width=\textwidth]{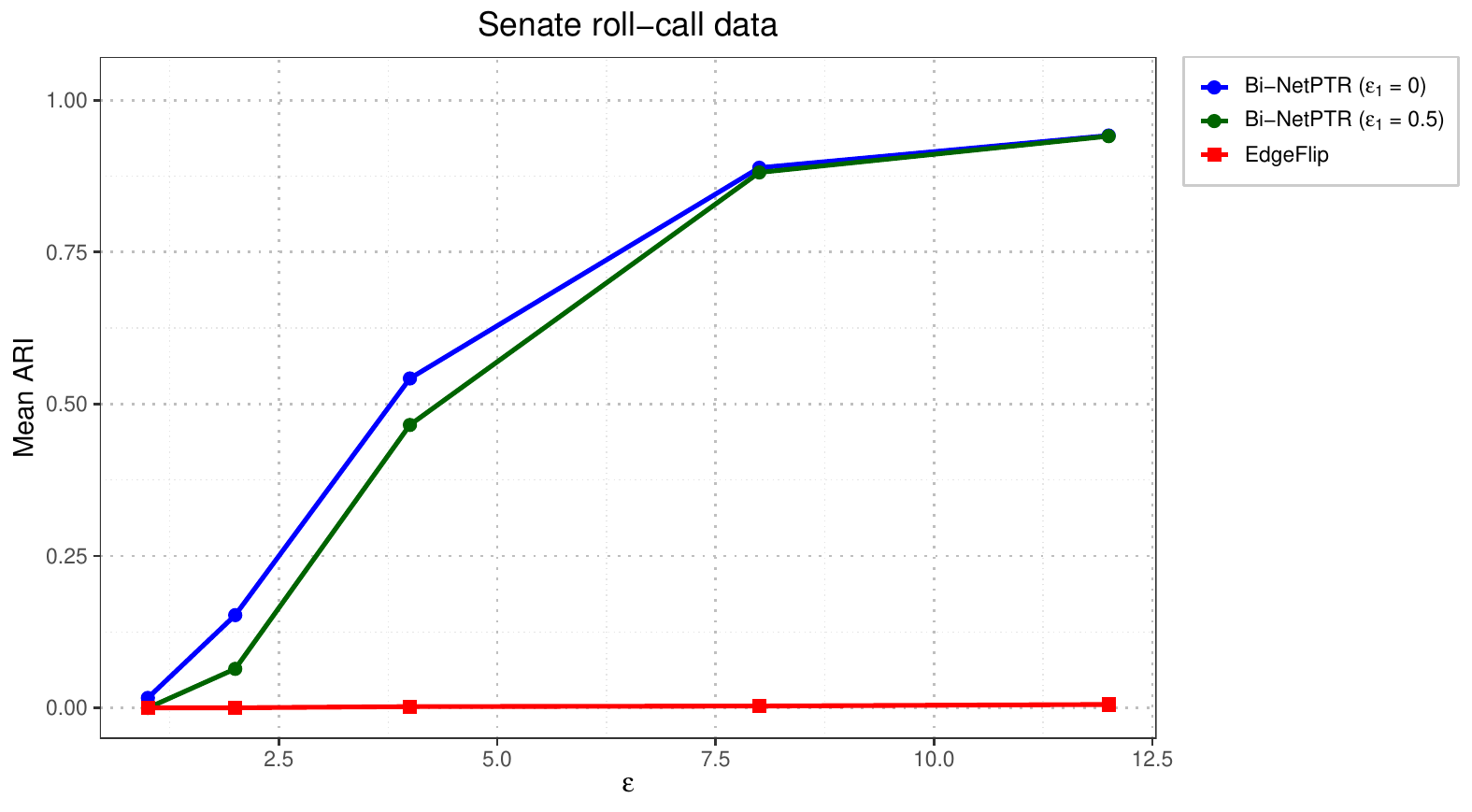}
    \caption*{(a) ARI against the non-private partition}
\end{minipage}
\hfill
\begin{minipage}[!htbp]{0.49\textwidth}
    \centering
    \includegraphics[width=\textwidth]{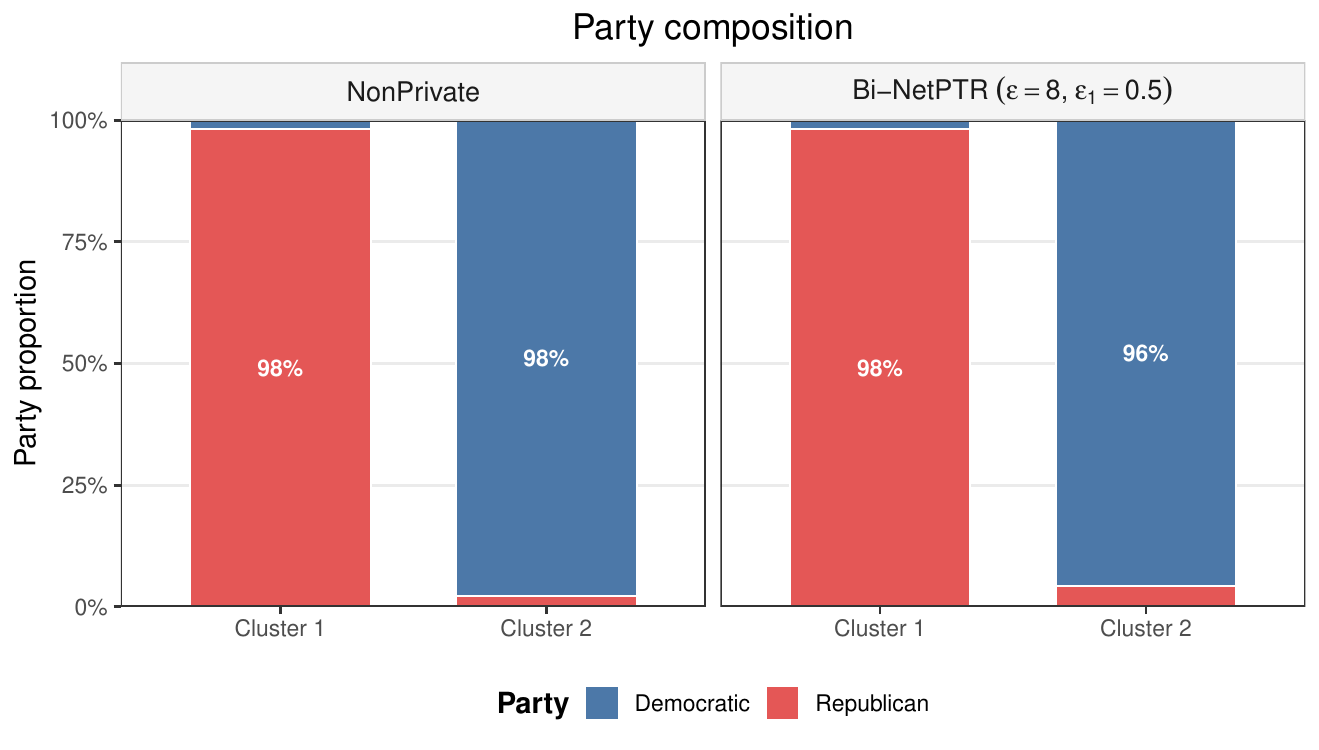}
    \caption*{(b) Party composition of recovered clusters}
\end{minipage}
\caption{Left panel: ARI between the private label estimate and the non-private label estimate versus privacy budgets. 
Right panel: the party composition within each recovered cluster for the non-private method and for Bi-NetPTR.}
\label{fig:senate_realdata}
\end{figure}

To understand how privacy protection manifests in this setting, we consider two test statistics $t_1$ and $t_2$. 
Let $t_1(j)$ be the two-sample proportion $z$-statistic on the yes-rate for $j$-th roll call between the overall Republican and Democratic, which follows standard normal distribution under null hypothesis. For $t_2(j)$, we again compare the yes-rates, but between the Republican legislators in Cluster 2 and the other Democrats. 
The histograms of $t_1$ and $t_2$ across roll-calls are summarized in Figure \ref{fig:inspect}. 
It can be found that $t_1$ exhibits a bimodal pattern, reflecting strong party distinction on many roll calls, whereas the $t_2$ values are concentrated near 0, indicating weak separation between Republican legislators assigned to the Democratic-dominated cluster and the Democratic legislators themselves. 
This pattern suggests that the privacy protection mainly affects less polarized Republicans, whose voting records depart more often from the Republican majority. By contrast, highly polarized legislators with consistently party-line voting records remain in the Republican-dominated cluster. 
In a politically polarized environment, the response profile to a single vote may itself be sensitive, because it reveals how the population divides on a contested issue. Bi-NetPTR protects this full roll-call profile while still allowing the release of aggregate community structure.
%Such weakly separating roll calls are more easily obscured by the node-private perturbation, thereby masking the individual-level contribution of these legislators while largely preserving the aggregate partisan structure.
%This also explains why misclustering primarily arises from legislators with such low-signal votes, rather than with strongly polarized votes.

\begin{figure}[!htbp]
    \centering
    \includegraphics[width=0.7\linewidth]{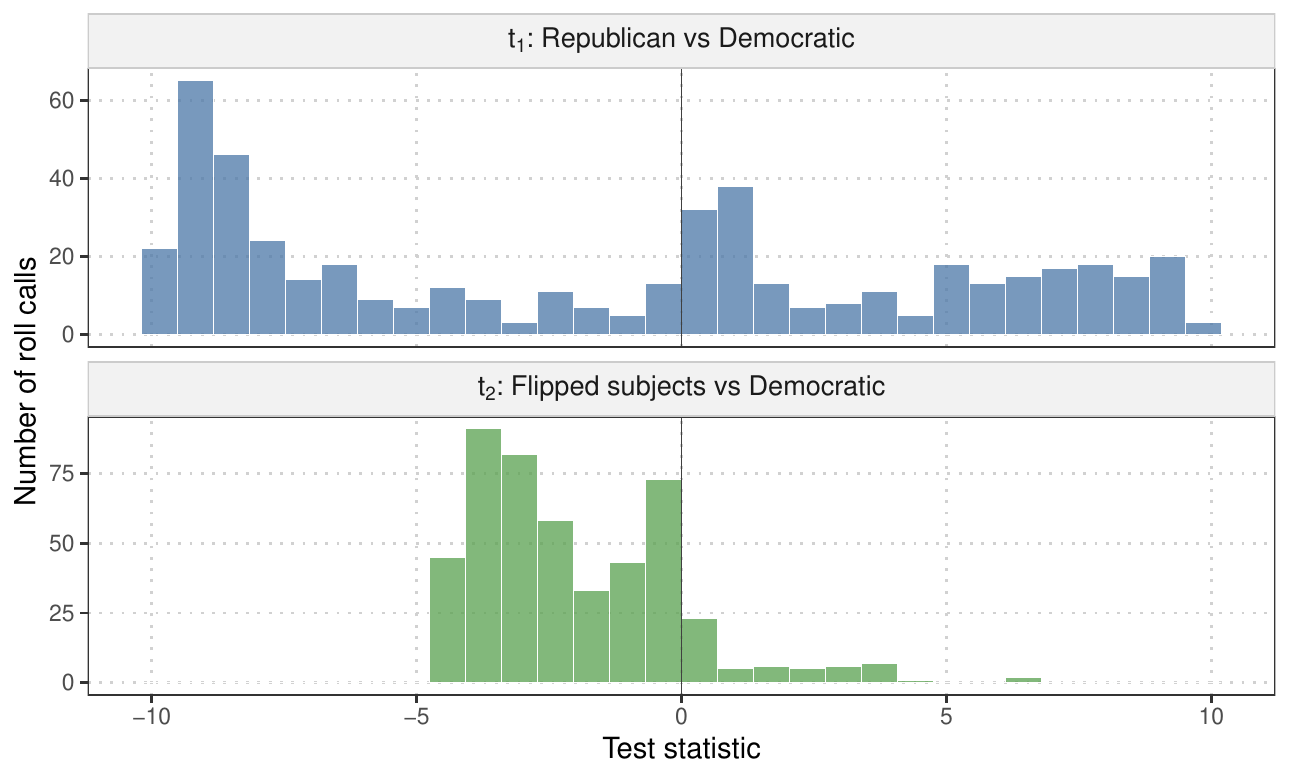}
  \caption{Roll-call diagnostic based on two test statistics.}
    \label{fig:inspect}
\end{figure}

\bigskip
\begin{center}
{\large\bf SUPPLEMENTARY MATERIAL}
\end{center}

\begin{description}

\item[Supplementary Materials:] Theoretical proofs of the main theorems and propositions, and additional numerical results on the simulated data and real data.
(pdf file)

\end{description}

\bibliographystyle{chicago}
\bibliography{ref}

@article{newman2013community,
  title={Community detection and graph partitioning},
  author={Newman, Mark EJ},
  journal={Europhysics Letters},
  volume={103},
  number={2},
  pages={28003},
  year={2013},
  publisher={EDP Sciences, IOP Publishing and Societ{\`a} Italiana di Fisica}
}

@article{rohe2011spectral,
  title={Spectral clustering and the high-dimensional stochastic blockmodel},
  author={Rohe, Karl and Chatterjee, Sourav and Yu, Bin},
  journal={The Annals of Statistics},
  volume={39},
  number={4},
  pages={596},
  year={2011},
  publisher={Institute of Mathematical Statistics}
}

@article{lei2015consistency,
  title={Consistency of spectral clustering in stochastic block models},
  author={Lei, Jing and Rinaldo, Alessandro},
  journal={The Annals of Statistics},
  pages={215--237},
  year={2015},
  publisher={JSTOR}
}

@article{DCSBM,
  title={Stochastic blockmodels and community structure in networks},
  author={Karrer, Brian and Newman, Mark EJ},
  journal={Physical Review E},
  volume={83},
  number={1},
  pages={016107},
  year={2011},
  publisher={APS}
}

@article{tong2025uniform,
  title={Uniform error bound for {PCA} matrix denoising},
  author={Tong, Xin T and Wang, Wanjie and Wang, Yuguan},
  journal={Bernoulli},
  volume={31},
  number={3},
  pages={2251--2275},
  year={2025},
  publisher={Bernoulli Society for Mathematical Statistics and Probability}
}

@article{SCORE,
  title={Fast community detection by {SCORE}},
  author={Jin, Jiashun},
  journal={The Annals of Statistics},
  volume={43},
  number={1},
  pages={57--89},
  year={2015},
  publisher={Institute of Mathematical Statistics}
}

@article{hu2024network,
  title={Network-adjusted covariates for community detection},
  author={Hu, Yaofang and Wang, Wanjie},
  journal={Biometrika},
  volume = {111}, 
  number = {4}, 
  pages={1221–1240},
  year={2024},
  publisher={Oxford University Press}
}

@article{shen2025optimal,
  title={Optimal Network-Guided Covariate Selection for High-Dimensional Data Integration},
  author={Shen, Tao and Wang, Wanjie},
  journal={arXiv preprint arXiv:2504.04866},
  year={2025}
}

@article{dwork2014algorithmic,
  title={The algorithmic foundations of differential privacy},
  author={Dwork, Cynthia and Roth, Aaron and others},
  journal={Foundations and trends{\textregistered} in theoretical computer science},
  volume={9},
  number={3--4},
  pages={211--407},
  year={2014},
  publisher={Now Publishers, Inc.}
}

@inproceedings{dwork2009differential,
  title={Differential privacy and robust statistics},
  author={Dwork, Cynthia and Lei, Jing},
  booktitle={Proceedings of the forty-first annual ACM symposium on Theory of computing},
  pages={371--380},
  year={2009}
}

@article{brunel2020propose,
  title={Propose, test, release: Differentially private estimation with high probability},
  author={Brunel, Victor-Emmanuel and Avella-Medina, Marco},
  journal={arXiv preprint arXiv:2002.08774},
  year={2020}
}

@article{shen2026efficient,
  title={Efficient Proposal-Test-Release for Minimax Optimal Estimation},
  author={Shen, Tao and Tong, Xin T and Wang, Wanjie},
  journal={arXiv preprint arXiv:2605.03264},
  year={2026}
}

@article{fan2018ell_,
  title={An $\ell_{\infty}$ Eigenvector Perturbation Bound and Its Application},
  author={Fan, Jianqing and Wang, Weichen and Zhong, Yiqiao},
  journal={Journal of Machine Learning Research},
  volume={18},
  number={207},
  pages={1--42},
  year={2018}
}

@article{daviskahan,
  title={The rotation of eigenvectors by a perturbation. III},
  author={Davis, Chandler and Kahan, William Morton},
  journal={SIAM Journal on Numerical Analysis},
  volume={7},
  number={1},
  pages={1--46},
  year={1970},
  publisher={SIAM}
}

@inproceedings{mohamed2022differentially,
  title={Differentially private community detection for stochastic block models},
  author={Mohamed, Mohamed S and Nguyen, Dung and Vullikanti, Anil and Tandon, Ravi},
  booktitle={International Conference on Machine Learning},
  pages={15858--15894},
  year={2022},
  organization={PMLR}
}

@article{li2023private,
  title={Private graph data release: A survey},
  author={Li, Yang and Purcell, Michael and Rakotoarivelo, Thierry and Smith, David and Ranbaduge, Thilina and Ng, Kee Siong},
  journal={ACM Computing Surveys},
  volume={55},
  number={11},
  pages={1--39},
  year={2023},
  publisher={ACM New York, NY}
}

@article{mueller2022sok,
  title={{SoK}: Differential Privacy on Graph-Structured Data},
  author={Mueller, Tamara T and Usynin, Dmitrii and Paetzold, Johannes C and Rueckert, Daniel and Kaissis, Georgios},
  journal={arXiv e-prints},
  pages={arXiv--2203},
  year={2022}
}

@inproceedings{hay2009accurate,
  title={Accurate estimation of the degree distribution of private networks},
  author={Hay, Michael and Li, Chao and Miklau, Gerome and Jensen, David},
  booktitle={2009 Ninth IEEE International Conference on Data Mining},
  pages={169--178},
  year={2009},
  organization={IEEE}
}

@article{nguyen2023faster,
  title={Faster approximate subgraph counts with privacy},
  author={Nguyen, Dung and Halappanavar, Mahantesh and Srinivasan, Venkatesh and Vullikanti, Anil},
  journal={Advances in Neural Information Processing Systems},
  volume={36},
  pages={70402--70432},
  year={2023}
}

@article{warner1965randomized,
  title={Randomized response: A survey technique for eliminating evasive answer bias},
  author={Warner, Stanley L},
  journal={Journal of the American statistical association},
  volume={60},
  number={309},
  pages={63--69},
  year={1965},
  publisher={Taylor \& Francis}
}

@article{karwa2017sharing,
  title={Sharing social network data: differentially private estimation of exponential family random-graph models},
  author={Karwa, Vishesh and Krivitsky, Pavel N and Slavkovi{\'c}, Aleksandra B},
  journal={Journal of the Royal Statistical Society Series C: Applied Statistics},
  volume={66},
  number={3},
  pages={481--500},
  year={2017},
  publisher={Oxford University Press}
}

@article{holohan2017optimal,
  title={Optimal differentially private mechanisms for randomised response},
  author={Holohan, Naoise and Leith, Douglas J and Mason, Oliver},
  journal={IEEE Transactions on Information Forensics and Security},
  volume={12},
  number={11},
  pages={2726--2735},
  year={2017},
  publisher={IEEE}
}

@article{hehir2022consistent,
  title={Consistent spectral clustering of network block models under local differential privacy},
  author={Hehir, Jonathan and Slavkovi{\'c}, Aleksandra and Niu, Xiaoyue},
  journal={The Journal of privacy and confidentiality},
  volume={12},
  number={2},
  pages={10--29012},
  year={2022}
}

@inproceedings{mulle2015privacy,
  title={Privacy-Integrated Graph Clustering Through Differential Privacy.},
  author={M{\"u}lle, Yvonne and Clifton, Chris and B{\"o}hm, Klemens},
  booktitle={EDBT/ICDT Workshops},
  volume={1330},
  pages={247--254},
  year={2015}
}

@inproceedings{nguyen2016detecting,
  title={Detecting Communities under Differential Privacy},
  author={Nguyen, Hiep H and Imine, Abdessamad and Rusinowitch, Micha{\"e}l},
  booktitle={Workshop on Privacy in the Electronic Society-WPES 206},
  pages={83--93},
  year={2016}
}

@article{mukherjee2025local,
title={Local Differential Privacy-Preserving Spectral Clustering for General Graphs},
author={Sayan Mukherjee and Vorapong Suppakitpaisarn},
journal={Transactions on Machine Learning Research},
issn={2835-8856},
year={2025},
url={https://openreview.net/forum?id=zo5b60AuAH},
}

@article{he2024common,
  title={Common neighborhood estimation over bipartite graphs under local differential privacy},
  author={He, Yizhang and Wang, Kai and Zhang, Wenjie and Lin, Xuemin and Zhang, Ying},
  journal={Proceedings of the ACM on Management of Data},
  volume={2},
  number={6},
  pages={1--26},
  year={2024},
  publisher={ACM New York, NY, USA}
}

@article{zhen2024consistent,
  title={Consistent community detection in multi-layer networks with heterogeneous differential privacy},
  author={Zhen, Yaoming and Xu, Shirong and Wang, Junhui},
  journal={arXiv preprint arXiv:2406.14772},
  year={2024}
}

@article{klopp2026node,
  title={Node-Private Community Detection in Stochastic Block Models},
  author={Klopp, Olga and Zadik, Ilias},
  journal={arXiv preprint arXiv:2604.09078},
  year={2026}
}

@article{marchis2026node,
  title={Node-private community estimation in stochastic block models: Tractable algorithms and lower bounds},
  author={Marchis, Laurentiu and D'souza, Ethan and Fl{\'\i}dr, Tom{\'a}{\v{s}} and Loh, Po-Ling},
  journal={arXiv preprint arXiv:2605.15943},
  year={2026}
}

@article{lewis2019voteview,
  title={Voteview: Congressional roll-call votes database},
  author={Lewis, Jeffrey B and Poole, Keith and Rosenthal, Howard and Boche, Adam and Rudkin, Aaron and Sonnet, Luke},
  journal={See https://voteview. com/(accessed 27 July 2018)},
  year={2019}
}

@article{jackman2015package,
  title={Package ‘pscl’},
  author={Jackman, Simon and Tahk, Alex and Zeileis, Achim and Maimone, Christina and Fearon, Jim and Meers, Zoe and Jackman, Maintainer Simon and Imports, MASS},
  journal={Political Science Computational Laboratory},
  volume={18},
  number={04.2017},
  year={2015}
}

@article{Zhao2025,
	Author = {Zhao, Da and Wang, Wanjie and Li, Jialiang},
	Doi = {10.1007/s11222-025-10661-3},
	Isbn = {1573-1375},
	Journal = {Statistics and Computing},
	Number = {5},
	Pages = {118},
	Title = {Spectral clustering on aggregated multilayer networks with covariates},
	Url = {https://doi.org/10.1007/s11222-025-10661-3},
	Volume = {35},
	Year = {2025}
    }

\end{document}